\documentclass[review]{elsarticle}

\usepackage{lineno,hyperref}
\usepackage{numcompress}
\biboptions{sort&compress}

\usepackage{amsmath,amssymb,amsthm}

\usepackage{enumitem}
\usepackage[font={footnotesize},labelfont={bf}]{caption}
\usepackage{xcolor,color,soul}

\usepackage{complexity}
\usepackage[english,ruled,vlined,linesnumbered]{algorithm2e}

\usepackage{changepage}
\usepackage{setspace}
\usepackage{subcaption}
\usepackage{multirow}
\usepackage{color}

\usepackage{booktabs}

\usepackage[labelfont= bf, font={small, stretch=1.0}]{caption}


\SetAlgoCaptionSeparator{}
\SetKwInput{KwInput}{Input}
\newlength{\commentWidth}
\newtheorem{definition}{Definition}
\newtheorem{theorem}{Theorem}
\newtheorem{proposition}[theorem]{Proposition}
\newtheorem{lemma}[theorem]{Lemma}
\newtheorem{corollary}[theorem]{Corollary}

\renewcommand{\vec}[1]{\mathbf{#1}}
\begin{document}

\begin{frontmatter}

\title{Refined Kolmogorov Complexity of Analog, Evolving and Stochastic Recurrent Neural Networks}
\author[address1,address2]{J\'er\'emie Cabessa}
\author[address1]{Yann Strozecki}
\address[address1]{
Laboratoire DAVID, UVSQ -- University Paris-Saclay, 78035 Versailles, France
\href{mailto:jeremie.cabessa@uvsq.f}{{\rm \texttt{jeremie.cabessa@uvsq.fr}}}
\href{mailto:yann.strozecky@uvsq.f}{{\rm \texttt{yann.strozecki@uvsq.fr}}}
\medskip }
\address[address2]{
Institute of Computer Science of the Czech Academy of Sciences, \\18207 Prague 8, Czech Republic 
}

%\ead{jeremie.cabessa@u-paris2.fr, aubin.tchaptchet@gmail.com}

\begin{abstract}
We provide a refined characterization of the super-Turing computational power of analog, evolving, and stochastic neural networks based on the Kolmogorov complexity of their real weights, evolving weights, and real probabilities, respectively. First, we retrieve an infinite hierarchy of classes of analog networks defined in terms of the Kolmogorov complexity of their underlying real weights. This hierarchy is located between the complexity classes $\mathbf{P}$ and $\mathbf{P/poly}$. Then, we generalize this result to the case of evolving networks. A similar hierarchy of Kolomogorov-based complexity classes of evolving networks is obtained. This hierarchy also lies between $\mathbf{P}$ and $\mathbf{P/poly}$. Finally, we extend these results to the case of stochastic networks employing real probabilities as source of randomness. An infinite hierarchy of stochastic networks based on the Kolmogorov complexity of their probabilities is therefore achieved. In this case, the hierarchy bridges the gap between $\mathbf{BPP}$ and $\mathbf{BPP/log^*}$. Beyond proving the existence and providing examples of such hierarchies, we describe a generic way of constructing them based on classes of functions of increasing complexity. For the sake of clarity, this study is formulated within the framework of echo state networks. Overall, this paper intends to fill the missing results and provide a unified view about the refined capabilities of analog, evolving and stochastic neural networks.
\end{abstract}

% Short description
% ...


\end{frontmatter}

%\linenumbers
\pagebreak

\section{Introduction}
\label{sec:introduction}

% PARAGRAPH %
Philosophical considerations aside, it can reasonably be claimed that several brain processes are of a computational nature. ``The idea that brains are computational in nature has spawned a range of explanatory hypotheses in theoretical neurobiology''~\cite{ChurchlandSejnowski16}. In this regard, the question of the computational capabilities of neural networks naturally arises, among many others. 
%The brain can count, make deductions, perform algorithmic-like tasks, and at a lower level, compile several kinds of stimuli into a neural code.
%Consequently, the following questions arise: How does the brain compute? How do neural networks encode and process information? Can neural networks implement abstract models of computations? What are the computational capabilities of neural networks? 

% PARAGRAPH %
Since the early 1940s, the theoretical approach to neural computation has been focused on comparing the computational powers of neural network models and abstract computing machines. In 1943, McCulloch and Pitts proposed a modeling of the nervous system as a finite interconnection of logical devices and studied the computational power of ``nets of neurons'' from a logical perspective~\cite{McCullochPitts43}. Along these lines, Kleene and Minsky proved that recurrent neural networks composed of McCulloch and Pitts (i.e., Boolean) cells are computationally equivalent to finite state automata~\cite{Kleene56,Minsky67}. 
%In Minsky's own words:
%\begin{quote}
%It is evident that each neural network of the kind we have been considering is a finite-state machine. [\dots] It is interesting and even surprising that there is a converse to this. Every finite-state machine is equivalent to, and can be simulated by, some neural net.~\cite{Minsky67}
%\end{quote}
These results paved the way for a future stream of research motivated by the expectation to implement abstract machines on parallel hardware architectures (see for instance~\cite{CleeremansEtAl89,AlonEtAl91,Kremer95,Indyk95,HorneHush96,OmlinGiles96a,Siegelmann96}). 
% Elman90, Pollack91, GilesEtAl92, WatrousKuhn92, ZengEtAl93, GoudreauEtAl94, AlquezarSanfeliu95, FrasconiEtAl95, FrasconiEtAl96, OmlinGiles96b

% PARAGRAPH %
In 1948, Turing introduced the B-type unorganized machine, a kind of neural network composed of interconnected NAND neuronal-like units~\cite{Turing48}. He suggested that the consideration of sufficiently large B-type unorganized machines could simulate the behavior of a universal Turing machine with limited memory. The Turing universality of neural networks involving infinitely many Boolean neurons has been further investigated (see for instance~in~\cite{Pollack87,HartleySzu87,GarzonFranklin89,FranklinGarzon90,Schmidhuber90}). Besides, Turing brilliantly anticipated the concepts of ``learning'' and ``training'' that would later become central to machine learning. These concepts took shape with the introduction of the {\it perceptron}, a formal neuron that can be trained to discriminate inputs using Hebbian-like learning~\cite{Rosenblatt57,Rosenblatt58,Hebb49}. But the computational limitations of the perceptron dampened the enthusiasms for artificial neural networks~\cite{MinskyPapert69}. The ensuing winter of neural networks lasted until the 1980s, when the 
%introduction of Kohonen's self-organizing maps~\cite{Kohonen82,Kohonen84}, Fukushima's neocognitron, and most of all, the 
popularization of the backpropagation algorithm, among other factors, paved the way for the great success of deep learning~\cite{RumelhartEtAl86,Schmidhuber15}. 
%From a theoretical perspective, feedforward neural networks turn out to have universal approximation capabilities: they can approximate any well-behaved continuous or Borel real function to any degree of accuracy \cite{HornikEtAl89,Hornik91}. 

% PARAGRAPH %
Besides, in the late 50's, von Neumann proposed an alternative approach to brain information processing from the hybrid perspective of digital and analog computations~\cite{vonNeumann58}. %He considered that the non-linear character of the operations performed by the brain emerges from a combination of discrete and continuous mechanisms. Accordingly, he envisioned neural computation as something strictly more powerful than abstract machines.
Along these lines, Siegelmann and Sontag studied the capabilities of {\it sigmoidal neural networks}, (instead of Boolean ones). They showed that recurrent neural networks composed of linear-sigmoid cells and rational synaptic weights are Turing complete~\cite{SiegelmannSontag95,Hyotyniemi96,Neto97}. This result has been generalized to a broad class of sigmoidal networks~\cite{KilianSiegelmann96}. %The computational equivalence between rational recurrent neural networks and Turing machines has now become a standard result in the field.

% PARAGRAPH %
Following the developments in analog computation~\cite{Siegelmann99}, Siegelmann and Sontag argued that the variables appearing in the underlying chemical and physical phenomena could be modeled by continuous rather than discrete (rational) numbers. Accordingly, they introduced the concept of an {\it analog neural network} -- a sigmoidal recurrent neural net equipped with real instead of rational weights. They proved that analog neural networks are computationally equivalent to Turing machines with advice, and hence, decide the complexity class $\mathbf{P/poly}$ in polynomial time of computation~\cite{SiegelmannSontag94,Siegelmann03}. Analog networks are thus capable of {\it super-Turing} capabilities and could capture chaotic dynamical features that cannot be described by Turing machines~\cite{Siegelmann95}. Based to these considerations, Siegelmann and Sontag formulated the so-called Thesis of Analog Computation -- an analogous to the Church-Turing thesis in the realm of analog computation -- stating that no reasonable abstract analog device can be more powerful than first-order analog recurrent neural networks~\cite{SiegelmannSontag94,Siegelmann99}.

% PARAGRAPH %
Inspired by the learning process of neural networks, Cabessa and Siegelmann studied the computational capabilities of evolving neural networks~\cite{CabessaSiegelmann11,CabessaSiegelmann14}. In summary, evolving neural networks using either rational, real, or binary evolving weights are all equivalent to analog neural networks. They also decide the class $\mathbf{P/poly}$ in polynomial time of computation. 

% PARAGRAPH %
The computational power of stochastic neural networks has also been investigated in detail. For rational-weighted networks, the addition of a discrete source of stochasticity increases the computational power from $\mathbf{P}$ to $\mathbf{BPP/log^*}$, while for the case of real-weighted networks, the capabilities remain unchanged to the $\mathbf{P/poly}$ level~\cite{Siegelmann99b}. On the other hand, the presence of analog noise would strongly reduce the computational power of the systems to that of finite state automata, or even below~\cite{Ben-HurEtAl04,MaassOrponen98,MaassSontag99}. 

% PARAGRAPH %
Based on these considerations, a refined approach to the computational power of recurrent neural networks has been undertaken. On the one hand, the sub-Turing capabilities of Boolean rational-weighted networks containing $0$, $1$, $2$ or $3$ additional sigmoidal cells have been investigated~\cite{Sima19,Sima20}. On the other hand, a refinement of the super-Turing computational power of analog neural networks has been described in terms of the Kolmogorov complexity of the underlying real weights~\cite{SiegelmannEtAl97}. The capabilities of analog networks with weights of increasing Kolmogorov complexity shall stratify the gap between the complexity classes $\mathbf{P}$ and $\mathbf{P/poly}$.

% PARAGRAPH %
The capabilities of analog and evolving neural networks have been generalized to the context of infinite computation, in connection with the attractor dynamics of the networks~\cite{CabessaVilla12,CabessaSiegelmann12,CabessaVilla13,CabessaVilla14,CabessaVilla14b,CabessaDuparc15,CabessaVilla15,CabessaDuparc16,CabessaVilla16}. In this framework, the expressive power of the networks is characterized in terms of topological classes from the Cantor space (the space of infinite bit streams). A refinement of the computational power of the networks based on the complexity of the underlying real and evolving weights has also been described in this context~\cite{CabessaFinkel17,CabessaFinkel19}. 

% PARAGRAPH %
The computational capabilities of {\it spiking neural networks} (instead of sigmoidal one) has also been extensively studied~\cite{Maass99,MaassBishop99}. In this approach, the computational states are encoded into the temporal differences between spikes rather than within the activation values of the cells. Maass proved that single spiking neurons are strictly more powerful than single threshold gates~\cite{MaassSchmitt97,MaassSchmitt99}. He also characterized lower and upper bounds on the complexity of networks composed of classical and noisy spiking neurons (see~\cite{Maass94,Maass96,Maass97a,Maass97b,MaassRuf99,MaassMarkram04} and~\cite{Maass95,Maass96b}, respectively). He further showed that networks of spiking neurons are capable of simulating analog recurrent neural networks~\cite{Maass98}.

% PARAGRAPH %
In the 2000s, P\u{a}un introduced the concept of a {\it P system} -- a highly parallel abstract model of computation inspired by the membrane-like structure of the biological cell~\cite{Paun00,Paun02}. His work led to the emergence of a highly active field of research. The capabilities of various models of so-called {\it neural P systems} have been studied (see for instance~\cite{PaunEtAl06,PaunEtAl06b,PaunEtAl07,Paun07,Paun08}). In particular, neural P systems provided with a bio-inspired source of acceleration were shown to be capable of hypercomputational capabilities, spanning all levels of the arithmetical hierarchy~\cite{PaunCalude04,GheorgheStannett12}.

% PARAGRAPH %
In terms of practical applications, recurrent neural networks are natural candidates for sequential tasks, involving time series or textual data for instance. Classical recurrent architectures, like LSTM and GRU, have been applied with great success in many situations~\cite{SchmidhuberEtAl17}. A $3$-level formal hierarchy of the sub-Turing expressive capabilities of these architectures, based on the notions of space complexity and rational recurrence, has been established~\cite{MerrillEtAl20}. Echo state networks are another kind of recurrent neural networks enjoying an increasing popularity due to their training efficiency~\cite{Jaeger01,Jaeger02,Jaeger04,LukoseviciusJaeger09}. The computational capabilities of echo state networks have been studied from the alternative perspective of universal approximation theorems~\cite{HornikEtAl89,Hornik91}. In this context, echo state networks are shown to be universal, in the sense of being capable of approximating different classes of filters of infinite discrete time signals~\cite{GrigoryevaOrtega18a,GrigoryevaOrtega18b,GononOrtega20,GononOrtega21}. These works fit within the field of functional analysis rather than computability theory.

% PARAGRAPH %
In this paper, we extend the refined Kolmogorov-based complexity of analog neural networks~\cite{SiegelmannEtAl97} to the cases of evolving and stochastic neural networks~\cite{CabessaSiegelmann14,Siegelmann99b}. More specifically, we provide a refined characterization of the super-Turing computational power of analog, evolving, and stochastic neural networks based on the Kolmogorov complexity of their real weights, evolving weights, and real probabilities, respectively. First, we retrieve an infinite hierarchy of complexity classes of analog networks defined in terms of the Kolmogorov complexity of their underlying real weights. This hierarchy is located between the complexity classes $\mathbf{P}$ and $\mathbf{P/poly}$. Using a natural identification between real numbers and infinite sequences of bits, we generalize this result to the case of evolving networks. Accordingly, a similar hierarchy of Kolomogorov-based complexity classes of evolving networks is obtained. This hierarchy also lies between $\mathbf{P}$ and $\mathbf{P/poly}$. Finally, we extend these results to the case of stochastic networks employing real probabilities as source of randomness. An infinite hierarchy of complexity classes of stochastic networks based on the Kolmogorov complexity of their real probabilities is therefore achieved. In this case, the hierarchy bridges the gap between $\mathbf{BPP}$ and $\mathbf{BPP/log^*}$. Beyond proving the existence and providing examples of such hierarchies, we describe a generic way of constructing them based on classes of functions of increasing complexity. Technically speaking, the separability between non-uniform complexity classes is achieved by means of a generic diagonalization technique, a result of interest per se which improves upon the previous approach~\cite{SiegelmannEtAl97}. For the sake of clarity, this study is formulated within the framework of echo state networks. Overall, this paper intends to fill the missing results and provide a unified view about the refined capabilities of analog, evolving and stochastic neural networks.
% we provide a refinement of the computational power of analog, evolving, and stochastic neural networks based on the Kolmogorov complexity of their underlying real weights, evolving weights, and real probabilities, respectively. For the case of analog networks, we retrieve the results from Balc\'azar et al.~\cite{SiegelmannEtAl97}, and provide an infinite Kolmogorov-based hierarchy of complexity classes of analog networks bridging the gap between $\mathbf{P}$ and $\mathbf{P/poly}$. These results are then naturally extended to the case of evolving networks, resulting in a similar hierarchy of complexity classes between $\mathbf{P}$ and $\mathbf{P/poly}$. Finally, an infinite hierarchy of complexity classes of stochastic networks based on the Kolmogorov complexity of their real probabilities is also described. This hierarchy lies between $\mathbf{BPP}$ and $\mathbf{BPP/log^*}$. To achieve these results, the separability between non-uniform complexity classes is achieved by means of a generic diagonalization technique, which constitutes an improvement over previous results~\cite{SiegelmannEtAl97}. For the sake of clarity, this study is formulated within the framework of echo state networks. Overall, this papers intends to bridge the gaps and provide a unified view of the computational power of analog, evolving and stochastic neural networks.

% PARAGRAPH %
This paper is organized as follows. Section~\ref{sec:related_works} describes the related works. Section~\ref{sec:prelim} provides the mathematical notions necessary to this study. Section~\ref{sec:RNNs} presents recurrent neural networks within the formalism of echo state networks. Section~\ref{sec:AES-NNs} introduces the different models of analog, evolving and stochastic recurrent neural networks, and establishes their tight relations to non-uniform complexity classes defined in terms of Turing machines with advice. Section~\ref{sec:results} provides the hierarchy theorems, which in turn, lead to the descriptions of strict hierarchies of classes of analog, evolving and stochastic neural network. Section~\ref{sec:discussion} offers some discussion and concluding remarks.

\section{Related Works}
\label{sec:related_works}

% PARAGRAPH %
Kleene and Misnky showed the equvalence between Boolean recurrent neural networks and finite state automata~\cite{Kleene56,Minsky67}. % 1
Siegelmann and Sontag proved the Turing universality of rational-weighted neural networks~\cite{SiegelmannSontag95}. % 2
Kilian and Siegelmann generalized the result to a broader class of sigmoidal neural networks~\cite{KilianSiegelmann96}. % 3
In connection with analog computation, Siegelmann and Sontag characterized the super-Turing capabilities of real-weighted neural networks~\cite{SiegelmannSontag95,Siegelmann95,Siegelmann03}. % 4
Cabessa and Siegelmann extended the result to evolving neural networks~\cite{CabessaSiegelmann14}. % 5
The computational power of various kinds of stochastic and noisy neural networks has been characterized~\cite{Siegelmann99,Ben-HurEtAl04,MaassOrponen98,MaassSontag99}. % 6
Sima graded the sub-Turing capabilites of Boolean networks containing $0$, $1$, $2$ or $3$ additional sigmoidal cells~\cite{Sima19,Sima20}. % 7
Balc\'azar et al.~hierarchized the super-Turing computational power of analog networks in terms of the Kolmogorov complexity of their underlying real weights~\cite{SiegelmannEtAl97}. % 8
Cabessa et al.~pursued the study of the computational capabilities of analog and evolving neural networks from the perspective of infinite computation~\cite{CabessaVilla12,CabessaSiegelmann12,CabessaVilla13,CabessaVilla14,CabessaVilla14b,CabessaDuparc15,CabessaVilla15,CabessaDuparc16,CabessaVilla16}. % 9

% PARAGRAPH %
Besides, the computational power of spiking neural networks has been extensively studied by Maass~\cite{Maass99,MaassBishop99,MaassSchmitt97,MaassSchmitt99,Maass94,Maass96,Maass97a,Maass97b,MaassRuf99,MaassMarkram04,Maass95,Maass96b,Maass98}. In addition, since the 2000s, the field of P systems, which involves neural P systems in particular, has been booming (see for instance~\cite{Paun00,Paun02,PaunEtAl06,PaunEtAl06b,PaunEtAl07,Paun07,Paun08}). The countless variations of proposed models are generally Turing complete.

% PARAGRAPH %
Regarding modern architectures, a hierarchy of the sub-Turing expressive power of LSTM and GRU neural networks has been established~\cite{MerrillEtAl20}. Furthermore, the universality of echo state networks has been studied from the perspective of universal approximation theorems~\cite{GrigoryevaOrtega18a,GrigoryevaOrtega18b,GononOrtega20,GononOrtega21}.

\section{Preliminaries}
\label{sec:prelim}

% PARAGRAPH %
The binary alphabet is denoted by $\Sigma = \{ 0,1 \}$, and the set of finite words, finite words of length $n$, infinite words, and finite or infinite words over $\Sigma$ are denoted by $\Sigma^*$, $\Sigma^n$, $\Sigma^\omega$, and $\Sigma^{\leq\omega}$, respectively. Given some finite or infinite word $w \in \Sigma^{\leq\omega}$, the $i$-th bit of $w$ is denoted by $w_i$, the sub-word from index $i$ to index $j$ is $w[i:j]$, and the length of $w$ is $|w|$, with $|w| = \infty$ if $w \in \Sigma^\omega$.

% PARAGRAPH %
A \textit{Turing machine (TM)} is defined in the usual way. A \textit{Turing machine with advice (TM/A)} is a TM provided with an additional advice tape and function $\alpha : \mathbb{N} \rightarrow \Sigma^*$. On every input $w \in \Sigma^n$ of length $n$, the machine first queries its advice function $\alpha(n)$, writes this word on its advice tape, and then continues its computation according to its finite program. The advice $\alpha$ is called {\it prefix} if $m \leq n$ implies that $\alpha(m)$ is a prefix of $\alpha(n)$, for all $m, n \in \mathbb{N}$. The advice $\alpha$ is called {\it unbounded} if the length of the successive advice words tends to infinity, i.e., if $\lim_{n \rightarrow \infty} | \alpha(n) | = \infty$.\footnote{Note that if $\alpha$ is not unbounded, then it can be encoded into the program of a TM, and thus, doesn't add any computational power to the TM model.} In this work, we assume that every advice $\alpha$ is prefix and unbounded, which ensures that $\lim_{n \rightarrow \infty} \alpha(n) \in \Sigma^\omega$ is well-defined. For any non-decreasing function $f : \mathbb{N} \rightarrow \mathbb{N}$, the advice $\alpha$ is said to be of size $f$ if $|\alpha(n)| = f(n)$, for all $n \in \mathbb{N}$. We let  $\mathrm{poly}$ be the set of univariate polynomials with integer coefficients and $\mathrm{log}$ be the set of functions of the form $n \rightarrow C\log(n)$ where $C\in \mathbb{N}$.
The advice $\alpha$ is called {\it polynomial} or {\it logarithmic} if it is of size $f \in \mathrm{poly}$ or $f \in \mathrm{log}$, respectively. A TM/A $\mathcal{M}$ equipped with some prefix unbounded advice is assumed to satisfy the following additional consistency property: for any input $w \in \Sigma^n$, $\mathcal{M}$ accepts $w$ using advice $\alpha(n)$ iff $\mathcal{M}$ accepts $w$ using advice $\alpha(n')$, for all $n' \geq n$. 

The class of languages decidable in polynomial time by some TM is $\mathbf{P}$. The class of languages decidable in time $t : \mathbb{N} \rightarrow \mathbb{N}$ by some TM/A with advice $\alpha$ is denoted by $\mathbf{TMA}[\alpha, t]$. Given some class of advice functions $\mathcal{A} \subseteq ({\Sigma^*})^\mathbb{N}$ and some class of time functions $\mathcal{T} \subseteq \mathbb{N}^\mathbb{N}$, we naturally define
$$
\mathbf{TMA} \left[ \mathcal{A}, \mathcal{T} \right] = \bigcup_{\alpha \in \mathcal{A}} \bigcup_{t \in \mathcal{T}} \mathbf{TMA} \left[ \alpha, t \right] .
$$
The class of languages decidable in polynomial time by some TM/A with polynomial prefix and non-prefix advice are $\mathbf{P/poly^*}$ and $\mathbf{P/poly}$, respectively. It can be noticed that $\mathbf{P/poly^*} = \mathbf{P/poly}$. %For logarithmic advice however, it is the case that $\mathbf{P/log^*} \subsetneq \mathbf{P/log}$. 

% PARAGRAPH %
A \textit{probabilistic Turing machine (PTM)} is a TM with two transition functions. At each computational step, the machine chooses one or the other transition function with probability $\tfrac{1}{2}$, independently from all previous choices, and updates its state, tapes' contents, and heads accordingly. A PTM $\mathcal{M}$ is assumed to be a decider, meaning that for any input $w$, all possible computations of $\mathcal{M}$ end up either in an accepting or in a rejecting state. Accordingly, the random variable corresponding to the decision ($0$ or $1$) that $\mathcal{M}$ makes at the end of its computation over $w$ is denoted by $\mathcal{M}(w)$. 
%Given some function $f : \mathbb{N} \rightarrow \mathbb{N}$, we say that $\mathcal{M}$ runs in time $f$ if for every $w \in \Sigma^*$, $\mathcal{M}$ halts in $f(|w|)$ steps regardless of the random choices it makes. 
Given some language $L \subseteq \Sigma^*$, we say that the PTM $\mathcal{M}$ decides $L$ in time $t : \mathbb{N} \rightarrow \mathbb{N}$ if, for every $w \in \Sigma^*$, $\mathcal{M}$ halts in $t(|w|)$ steps regardless of its random choices, and $Pr[\mathcal{M}(w) = 1] \geq \tfrac{2}{3}$ if $w \in L$ and $Pr[\mathcal{M}(w) = 0] \geq \tfrac{2}{3}$ if $w \not \in L$. The class of languages decidable in polynomial time by some PTM is $\mathbf{BPP}$. A \textit{probabilistic Turing machine with advice (PTM/A)} is a PTM provided with an additional advice tape and function $\alpha : \mathbb{N} \rightarrow \Sigma^*$. The class of languages decided in time $t$ by some PTM/A with advice $\alpha$ is denoted by $\mathbf{PTMA}[\alpha, t]$. Given some class of advice functions $\mathcal{A} \subseteq ({\Sigma^*})^\mathbb{N}$ and some class of time functions, we also define
$$
\mathbf{PTMA} \left[ \mathcal{A}, \mathcal{T} \right] = \bigcup_{\alpha \in \mathcal{A}} \bigcup_{t \in \mathcal{T}} \mathbf{PTMA} \left[ \alpha, t \right] .
$$
The class of languages decidable in polynomial time by some PTM/A with logarithmic prefix and non-prefix advice are $\mathbf{BPP/log^*}$ and $\mathbf{BPP/log}$, respectively. In this probabilistic case however, it can be shown that $\mathbf{BPP/log^*} \subsetneq \mathbf{BPP/log}$. 

% PARAGRAPH %
In the sequel, we will be interested in the size of the advice functions. Hence, we define the following {\it non-uniform complexity classes}.\footnote{This definition is non-standard. Usually, non-uniform complexity classes are defined with respect to a class of advice functions $\mathcal{H} \subseteq (\Sigma^*)^{\mathbb{N}}$ instead of a class of advice functions' size $\mathcal{F} \subseteq \mathbb{N}^{\mathbb{N}}$.} 
%(of the form $\mathbf{P/poly^*}$ and $\mathbf{BPP/log^*}$) can be defined in an alternative way. 
Given a class of languages (or associated machines) $\mathcal{C} \subseteq 2^{\Sigma^*}$ and 
%a class of functions $\mathcal{F} \subseteq \mathbb{N}^{\mathbb{N}}$, 
a function $f : \mathbb{N} \rightarrow \mathbb{N}$, 
we say that 
%$L \in \mathcal{C} / \mathcal{F}^*$ 
$L \in \mathcal{C} / f^*$
if there exist some $L' \in \mathcal{C}$
%, some $f \in \mathcal{F}$, 
and some prefix advice function $\alpha : \mathbb{N} \rightarrow \Sigma^*$ such that, for all $n \in \mathbb{N}$ and for all $w \in \Sigma^n$, the following properties hold:
\begin{enumerate}[label=(\roman*),itemsep=0pt]
    \item $| \alpha(n) | = f(n)$, for all $n \geq 0$;
    \item $w \in L \Leftrightarrow \langle w, \alpha(n) \rangle \in L'$;
    \item $\langle w, \alpha(n) \rangle \in L' \Leftrightarrow \langle w, \alpha(k) \rangle \in L', \text{ for all } k \geq n$.
\end{enumerate}
Given a class of functions $\mathcal{F} \subseteq \mathbb{N}^{\mathbb{N}}$, we naturally set
$$
\mathcal{C} / \mathcal{F}^* = \bigcup_{f \in \mathcal{F}} \mathcal{C} / f^* .
$$
The non-starred complexity classes $\mathcal{C} / f$ and $\mathcal{C} / \mathcal{F}$ are defined analogously, except that the prefix property of $\alpha$ and the last condition are not required. 
%Taking $\mathcal{C} = \mathbf{P}$ and $\mathcal{H} = \mathbf{poly}$ yields the class $\mathbf{P/poly^*} = \mathbf{P/poly}$, while taking $\mathcal{C} = \mathbf{BPP}$ and $\mathcal{H} = \mathbf{log}$ leads to the class $\mathbf{BPP/log^*}$. 
For instance, the class of languages decidable in polynomial time by some Turing machine (resp.~probabilistic Turing machines) with prefix advice of size $f$ is $\mathbf{P} / f^*$ (resp.~$\mathbf{BPP} / f^*$).

% PARAGRAPH %
Besides, for any $w = w_0 w_1 w_2 \cdots \in \Sigma^{\leq \omega}$, we consider the {\it base-2} and {\it base-4 encoding} functions $\delta_2 : \Sigma^{\leq \omega} \rightarrow [0, 1]$ and $\delta_4 : \Sigma^{\leq \omega} \rightarrow [0, 1]$ respectively defined by
\begin{equation*}
\delta_2(w) = \sum_{i = 0}^{|w|} \frac{w_i + 1}{2^{i+1}} \text{ ~and~ } \delta_4(w) = \sum_{i = 0}^{|w|} \frac{2 w_i + 1}{4^{i+1}} .
\end{equation*}
The use of base $4$ ensures that $\delta_4$ is an injection. Setting $\Delta := \delta_4(\Sigma^\omega) \subseteq [0, 1]$ ensures that the restriction $\delta_4 : \Sigma^{\leq \omega} \rightarrow \Delta$ is bijective, and thus that $\delta^{-1}$ is well-defined on the domain $\Delta$. In the sequel, for any real $r \in \Delta$, its base-4 expansion will be generally be denoted as $\bar r = r_0 r_1 r_2 \cdots \in \delta_4^{-1}(r) \in \Sigma^\omega$. For any $R \subseteq \Delta$, we thus define $\bar R = \{ \bar r : r \in R \}$.

% PARAGRAPH %
%We now define the {\it Kolmogorov complexity} of a real number as stated in a related work~\cite{SiegelmannEtAl97}. Let $\mathcal{M}_U$ be a universal Turing machine, $f, g : \mathbb{N} \rightarrow \mathbb{N}$ be two functions and $\alpha \in \Sigma^\omega$ be some infinite word. We say that $\alpha \in \bar K^f_g$ if there exists $\beta \in \Sigma^{\omega}$ such that, for all but finitely many $n$, the machine $\mathcal{M}_U$ with inputs $\beta[0:m-1]$ and $n$ will output $\alpha[0:n-1]$ in time $g(n)$, for all $m \geq f(n)$. In other words, $\alpha \in \bar K^f_g$ if its $n$ first bits can be recovered from $f(n)$ first bits of some $\beta$ in time $g(n)$. The notion expressed its interest when $f(n) \leq n$, in which case $\alpha \in \bar K^f_g$ means that every $n$-long prefix of $\alpha$ can be compressed into and recovered from a smaller $f(n)$-long prefix of $\beta$. Given two classes of functions $\mathcal{F}$ and $\mathcal{G}$, we define $\bar K^{\mathcal{F}}_{\mathcal{G}} = \bigcup_{f \in \mathcal{F}} \bigcup_{g \in \mathcal{G}} \bar K^f_g$. Finally, for any real number $r \in \Delta$ with associated infinite word $\bar r = \delta^{-1}_4(r) \in \Sigma^{\omega}$, we say that $r \in K^f_g$ (resp.~$r \in K^{\mathcal{F}}_{\mathcal{G}}$) iff $\bar r \in \bar K^f_g$ (resp.~$\bar r \in \bar K^{\mathcal{F}}_{\mathcal{G}}$). 

% PARAGRAPH %
Finally, for every probability space $\Omega$ and every events $A_1, \dots, A_n \subseteq \Omega$, the probability that at least one event $A_i$ occurs can be bounded by the {\it union bound} defined by
\begin{equation*}
    \mathrm{Pr} \left( \bigcup_{i=1}^n A_i \right) \leq \sum_{i=1}^n \mathrm{Pr} \left( A_i \right) .
\end{equation*}

%\section{Models}
%\label{sec:models}

% PARAGRAPH %
%\subsection{Rational-Weighted Recurrent Neural Networks}
\section{Recurrent Neural Networks}
\label{sec:RNNs}

% PARAGRAPH %
We consider a specific model of recurrent neural networks complying with the {\it echo state networks} architecture~\cite{Jaeger01,Jaeger02,Jaeger04,LukoseviciusJaeger09}. More specifically, a recurrent neural network is composed of an input layer, a pool of interconnected neurons sometimes referred to as the {\it reservoir}, and an output layer, as illustrated in Figure~\ref{fig_RNN}. The networks read and accept or reject finite words over the alphabet $\Sigma = \{ 0,1 \}$ using a input-output encoding described below. Accordingly, they are capable of performing decisions of formal languages.

% PARAGRAPH %
\begin{definition}
\label{RNN:def}
A {\em rational-weighted recurrent neural network (RNN)} is a tuple
$$
\mathcal{N} = \left( \vec{x}, \vec{h}, \vec{y}, \vec{W_{in}},\vec{W_{res}}, \vec{W_{out}}, \vec{h}^0 \right)
$$
where
\begin{itemize}[itemsep=0pt]
    \item $\vec{x} = ( x_0, x_1 )$ is a sequence of two {\it input cells}, the {\it data input} $x_0$ and the {\it validation input} $x_1$;
    \item $\vec{h} = ( h_0,\dots,h_{K-1} )$ is a sequence of $K$ {\it hidden cells}, sometimes referred to as the {\it reservoir};
    \item $\vec{y} = ( y_0, y_1 )$ is a sequence of two {\it output cells}, the {\it data output} $y_0$ and the {\it validation output} $y_1$;
    \item $\vec{W_{in}} \in \mathbb{Q}^{K \times (2+1)}$ is a matrix of {\it input weights} and {\it biases}, where $w_{ij}$ is the weight from input $x_j$ to cell $h_i$, for $j \neq 2$, and $w_{i2}$ is the bias of $h_i$; 
    \item $\vec{W_{res}} \in \mathbb{Q}^{K \times K}$ is a matrix of {\it internal weights}, where $w_{ij}$ is the weight from cell $h_j$ to cell $h_i$; 
    \item $\vec{W_{out}} \in \mathbb{Q}^{2 \times K}$ is a matrix of {\it output weights}, where $w_{ij}$ is the weight from cell $h_j$ to output $y_i$; 
    \item $\vec{h}^0 = (h^0_0, \dots, h^0_{K-1}) \in [0, 1]^K$ is the {\it initial state} of $\mathcal{N}$, where each component $h_i^0$ is the initial activation value of cell $h_i$.
\end{itemize}
\end{definition}

% PARAGRAPH %
The {\it activation value} of the cell $x_i$, $h_j$ and $y_k$ at time $t$ is denoted by $x_i^t \in \{ 0, 1 \}$, $h_j^t \in [0, 1]$ and $y_k^t \in \{ 0, 1 \}$, respectively. Note that the activation values of input and output cells are Boolean, as opposed to those of the hidden cells. The {\it input}, {\it output} and {\it (hidden) state} of $\mathcal{N}$ at time $t$ are the vectors
$$
\vec{x}^t = (x_0^t, x_1^t) \in \mathbb{B}^2, ~~~ \vec{y}^t = (y_0^t, y_1^t) \in \mathbb{B}^2 ~\text{ and }~ \vec{h}^t = (h_0^t, \dots, h_{K-1}^t) \in [0, 1]^K \,,
$$
respectively. Given some input $\vec{x}^t$ and state $\vec{h}^t$ at time $t$, the state $\vec{h}^{t+1}$ and the output $\vec{y}^{t+1}$ at time $t+1$ are computed by the following equations:
\begin{eqnarray}
    \vec{h}^{t+1} & = & \sigma \left( \vec{W_{in}} (\vec{x}^t : 1) + \vec{W_{res}} \vec{h}^t \right) \label{rnn:eq1} \\ 
    \vec{y}^{t+1} & = & \theta \left( \vec{W_{out}} \vec{h}^{t+1} \right) \label{rnn:eq2}
\end{eqnarray}
where $(\vec{x}^t : 1)$ denotes the vector $(x_0^t, x_1^t, 1)$, and $\sigma$ and $\theta$ are the {\it linear sigmoid} and the {\it hard-threshold} functions respectively given by
$$
\sigma(x) = 
\begin{cases}
    0 & \text{if } x < 0 \\
    x & \text{if } 0 \leq x \leq 1 \\
    1 & \text{if } x > 1
\end{cases}
\text{ ~~and~~ }
\theta(x) = 
\begin{cases}
    0 & \text{if } x < 0 \\
    1 & \text{if } x \geq  1
\end{cases}
.
$$
The constant value $1$ in the input vector $(\vec{x}^t, 1)$ ensure that the hidden cells $\vec{h}$ receive the last column of $\vec{W_{in}} \in \mathbb{Q}^{K \times (2+1)}$ as biases at each time step $t \geq 0$. In the sequel, the bias of cell $h_i$ will be denoted as $w_{i2}$.

\begin{figure}[h!]
    \centering
    \includegraphics[width=12cm]{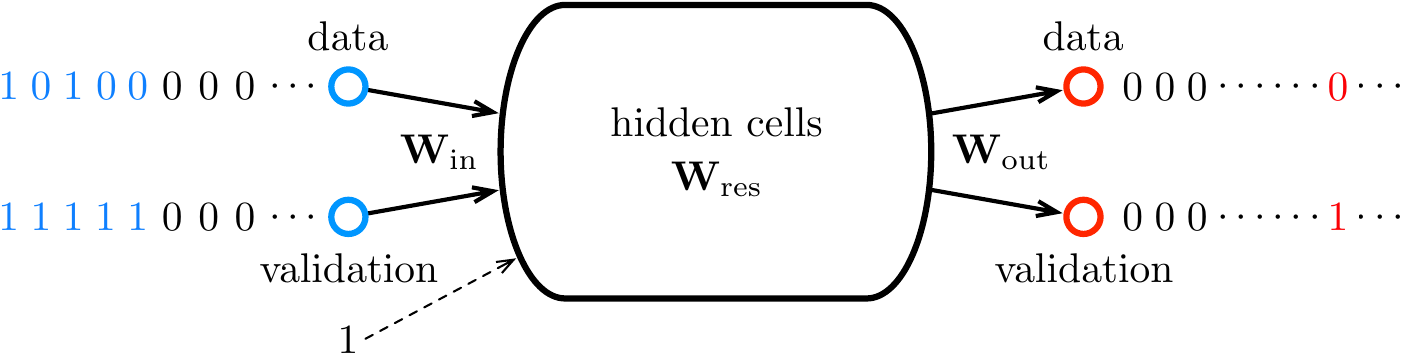}
    \caption{A recurrent neural network. The network is composed of two Boolean input cells (data and validation), two Boolean output cells (data and validation) and a set of $K$ hidden cells, the reservoir, that are recurrently interconnected. The weight matrices $\vec{W_{in}}$, $\vec{W_{res}}$, and $\vec{W_{out}}$ labeling the connections between these layers are represented. In this illustration, the network reads the finite word $10100$ by means of its data and validation input cells (blue), and eventually rejects it, as shown by the pattern of the data and validation output cells (red).}
    \label{fig_RNN}
\end{figure}

% PARAGRAPH %
An {\it input} $\vec{x}$ of length $n$ for the network $\mathcal{N}$ is an infinite sequence of inputs at successive time steps $t=0, 1, 2, \dots$, such that the $n$ first validation bits are equal to $1$, while the remaining data and validation bits are equal to $0$, i.e.,
$$
\vec{x} = \vec{x}^0\vec{x}^1 \cdots \vec{x}^{n-1} \vec{0}^\omega \in \left( \mathbb{B}^2 \right)^\omega
$$
where $\vec{x}^i = (x_0^i, 1)$ and $x_0^i \in \{ 0,1 \}$, for $i = 1, \dots, n-1$, and $\vec{0} = (0, 0)$. Suppose that the network $\mathcal{N}$ is in the initial state $\vec{h}^0$ and that input $\vec{x}$ is presented to $\mathcal{N}$ step by step. The dynamics given by Equations~(\ref{rnn:eq1}) and (\ref{rnn:eq2}) ensures that that $\mathcal{N}$ will generate the sequences of states and outputs
\begin{eqnarray*}
    \vec{h} & = & \vec{h}^0\vec{h}^1\vec{h}^2 \cdots \in \left( [0, 1]^K \right)^\omega \\
    \vec{y} & = & \vec{y}^1\vec{y}^2\vec{y}^3 \cdots \in \left( \mathbb{B}^2 \right)^\omega
\end{eqnarray*}
step by step, where $\vec{y} := \mathcal{N}(\vec{x})$ is the {\it output} of $\mathcal{N}$ associated with input $\vec{x}$.

% PARAGRAPH %
Now, let $w = w_0w_1 \cdots w_{n-1} \in \Sigma^*$ be some finite word of length $n$, let $\tau \in \mathbb{N}$ be some integer, and let $f : \mathbb{N} \rightarrow \mathbb{N}$ be some non-decreasing function. The word $w \in \Sigma^*$ can naturally be associated with the input 
$$
\vec{w} = \vec{w}^0\vec{w}^1 \cdots \vec{w}^{n-1} \vec{0}^\omega 
\in \left( \mathbb{B}^2 \right)^\omega
$$
defined by $\vec{w}^i = (x^i_0, x^i_1) = (w_i, 1)$, for $i = 1, \dots, n-1$. The input pattern is thus given by
\arraycolsep=2pt
$$
\begin{array}{c c c c c c c c c c c c c}
    x^0_0 & x^1_0 & x^2_0 & \cdots & = &
    w_0 & w_1 & \cdots & w_{n-1} & 0 & 0 & 0 & \cdots \\
    x^0_1 & x^1_1 & x^2_1& \cdots & = & 
    1 & 1 & \cdots & 1 & 0 & 0 & 0 & \cdots
\end{array}
$$
where the validation bits indicate whether an input is actively being processed or not, and the corresponding data bits represent the successive values of the input (see Figure~\ref{fig_RNN}). The word $w$ is said to be {\it accepted} or {\it rejected by $\mathcal{N}$ in time $\tau$} if the output 
$$
\mathcal{N}(\vec{w}) = \vec{y} = \vec{y}^1\vec{y}^2 \cdots \vec{y}^{\tau} \vec{0}^\omega \in \left( \mathbb{B}^2 \right)^\omega
$$
is such that $\vec{y^{i}} = (0, 0)$, for $i = 1, \dots, \tau-1$, and $\vec{y}^{\tau} = (1, 1)$ or $\vec{y}^{\tau} = (0, 1)$, respectively. The output pattern is thus given by
\arraycolsep=2pt
$$
\begin{array}{c c c c c c c c c c c c}
    y^0_0 & y^1_0 & y^2_0 & \cdots & = &
    0 & 0 & \cdots & y^{\tau}_1 & 0 & 0 & \cdots \\
    y^0_1 & y^1_1 & y^2_1 & \cdots & = & 
    0 & 0 & \cdots & 1 & 0 & 0 & \cdots
\end{array}
$$
where $y^{\tau}_1 = 1$ or $y^{\tau}_1 = 0$, respectively. 
%Note that, if $w$ is accepted or rejected in time $\tau$, then it is also accepted or rejected in time $\tau'$, respectively, for all $\tau' \geq \tau$. The word $w$ is said to be {\it rejected} by $\mathcal{N}$ in time $\tau$ if the output $\vec{y}$ satisfies $\vec{y^{i}} = (0, 0)$, for $i = 1, \dots, \tau-1$, and $\vec{y}^{\tau} = (0, 1)$. 

% PARAGRAPH %
In addition, the word $w$ is said to be {\it accepted} or {\it rejected by $\mathcal{N}$ in time $f$} if it is accepted or rejected in time $\tau \leq f(n)$, respectively.\footnote{The choice of the letter $f$ (instead of $t$) for referring to a computation time is deliberate, since the computation time of the networks will later be linked to the advice length of the Turing machines.} 
%In other words, the acceptance or rejection of input $w$ is determined by the pattern of the output $\vec{y}$. 
%The first moment when the validation bit turns to $1$ corresponds to the time step of the answer, and the answer is given by the data bit at this time (see Figure~\ref{fig_RNN}). 
A language $L \subseteq \Sigma^*$ is {\it decided by $\mathcal{N}$ in time $f$} if for every word $w \in \Sigma^*$, 
\begin{align*}
    w \in L & \text{ implies that } \vec{w} \text{ is accepted by $\mathcal{N}$ in time $f$ and} \\
w \not \in L & \text{ implies that } \vec{w} \text{ is rejected by $\mathcal{N}$ in time $f$}.
\end{align*}
%Sometimes, for the sake of clarity, we will say that $L \subseteq \Sigma^*$ is decided in time $f(n)$ (instead of simply $f$), where $n$ is the length of the input. 
A language $L \subseteq \Sigma^*$ is {\it decided by $\mathcal{N}$ in polynomial time} if there exists a polynomial $f$ such that $L$ is decided by $\mathcal{N}$ in time $f$. If it exists, the {\it language decided by $\mathcal{N}$ in time $f$} is denoted by $L_f(\mathcal{N})$. Besides, a network $\mathcal{N}$ is said to be a {\it decider} if any finite word is eventually accepted or rejected by it. In this case, 
%for any language $L \subseteq \Sigma^*$, there exists some function $g : \mathbb{N} \rightarrow \mathbb{N}$ such that $L_g(\mathcal{N}) = L$. If $\mathcal{N}$ is a decider, 
the {\it language decided by $\mathcal{N}$} is unique and is denoted by $L(\mathcal{N})$. We will assume that all the networks that we consider are deciders.
s

% PARAGRAPH %
Recurrent neural networks with rational weights have been shown to be computationally equivalent to Turing machines~\cite{SiegelmannSontag95}.

\begin{theorem}
\label{RRNN_thm1}
Let $L \subseteq \Sigma^*$ be some language. The following conditions are equivalent:
\begin{enumerate}[label=(\roman*),itemsep=0pt]
    %\item $L$ is recursive; \label{RRNN_thm1_c1}
    \item $L$ is decidable by some TM; \label{RRNN_thm1_c2}
    \item $L$ is decidable by some RNN. \label{RRNN_thm1_c3}
\end{enumerate}
\end{theorem}

\begin{proof}{(sketch)}
% new implication % 
%\ref{RRNN_thm1_c3} $\rightarrow$ \ref{RRNN_thm1_c1}: 
%The dynamics of any RNN $\mathcal{N}$ is governed by Equations (\ref{rnn:eq1}) and (\ref{rnn:eq2}), which are clearly recursive when the weight matrices are rational. 
% new implication % s
%\ref{RRNN_thm1_c1} $\rightarrow$ \ref{RRNN_thm1_c2}: 
%xThis implication holds by definition.
\ref{RRNN_thm1_c3} $\rightarrow$ \ref{RRNN_thm1_c2}: 
The dynamics of any RNN $\mathcal{N}$ is governed by Equations (\ref{rnn:eq1}) and (\ref{rnn:eq2}), which involves only rational weights, and thus can clearly be simulated by some Turing machine $\mathcal{M}$. 

% new implication % 
\ref{RRNN_thm1_c2} $\rightarrow$ \ref{RRNN_thm1_c3}: We provide a sketch of the original proof of this result~\cite{SiegelmannSontag95}.  This proof is based on the fact that any finite (and infinite) binary word can be encoded into and the activation value of a neuron, and decoded from this activation value bit by bit. This idea will be reused in some forthcoming proofs. First of all, recall that any TM $\mathcal{M}$ is computationally equivalently to, and can be simulated in real time by, some $p$-stack machine $\mathcal{S}$ with $p \geq 2$. We thus show that any $p$-stack machine $\mathcal{S}$ can be simulated by some RNN $\mathcal{N}$. Towards this purpose, we encode every stack content
$$w = w_0 \cdots w_{n-1} \in \{0,1\}^*$$
as the rational number
$$q_w = \delta_4(w) = \sum_{i=0}^{n-1} \frac{2 \cdot w(i) + 1}{4^{i+1}} \in [0,1].$$
For instance, $w = 1110$ is encoded into $q_w = \frac{3}{4} + \frac{3}{16} + \frac{3}{64} + \frac{1}{256}$. With this base-4 encoding, the required stack operations can be performed by simple functions involving the sigmoid-linear function $\sigma$, as described below:
\begin{itemize}[itemsep=0pt]
\item Reading the top of the stack: $\mathrm{top}(q_w) = \sigma(4q_w - 2)$
\item Pushing $0$ into the stack: $\mathrm{push_0}(q_w) = \sigma(\frac{1}{4}q_w + \frac{1}{4})$
\item Pushing $1$ into the stack: $\mathrm{push_1}(q_w) = \sigma(\frac{1}{4}q_w + \frac{3}{4})$
\item Popping the stack: $\mathrm{pop}(q_w) = \sigma(4q_w - (2 \mathrm{top}(q_w) + 1))$
\item Emptiness of the stack: $\mathrm{empty}(q_w) = \sigma(4q_w)$
\end{itemize}
%For instance, if $w = 1110$, then $q_w = \frac{3}{4} + \frac{3}{16} + \frac{3}{64} + \frac{1}{256}$ and one has:
%\begin{itemize}[itemsep=0pt]
%\item $top(q_w) = 1$
%\item $push_0(q_w) = \frac{1}{4} + \frac{3}{16} + \frac{3}{64} + \frac{3}{256} + \frac{1}{1024}$
%\item $push_1(q_w) = \frac{3}{4} + \frac{3}{16} + \frac{3}{64} + \frac{3}{256} + \frac{1}{1024}$
%\item $pop(q_w) = \frac{3}{4} + \frac{3}{16} + \frac{1}{64}$
%\item $empty(q_w) = 1$ (meaning that $w$ is not empty).
%\end{itemize}
Hence, the content $w$ of each stack can be encoded into the rational activation value $q_w$ of a stack neuron, and the stack operations (reading the top, pushing $0$ or $1$, popping and testing the emptiness) can be performed by simple neural circuits implementing the functions described above.

Based on these considerations, we can design an RNN $\mathcal{N}$ which correctly simulates the $p$-stack machine $\mathcal{S}$. The network $\mathcal{N}$ contains $3$ neurons per stack: one for storing the encoded content of the stack, one for reading the top element of the stack, and one for storing the answer of the emptiness test of the stack. Moreover, $\mathcal{N}$ contains two pools of neurons implementing the computational states and transition function of $\mathcal{S}$, respectively. For any computational state, input bit, and contents of the stacks, $\mathcal{N}$ computes the next computational state and updates the stacks' contents in accordance with the transition function of $\mathcal{S}$. In this way, the network $\mathcal{N}$ simulates the behavior of the $p$-stack machine $\mathcal{S}$ correctly. The network $\mathcal{N}$ is illustrated in Figure \ref{RNN_proof_fig}. 
\end{proof}

\begin{figure}
\begin{center}
\includegraphics[width=10cm]{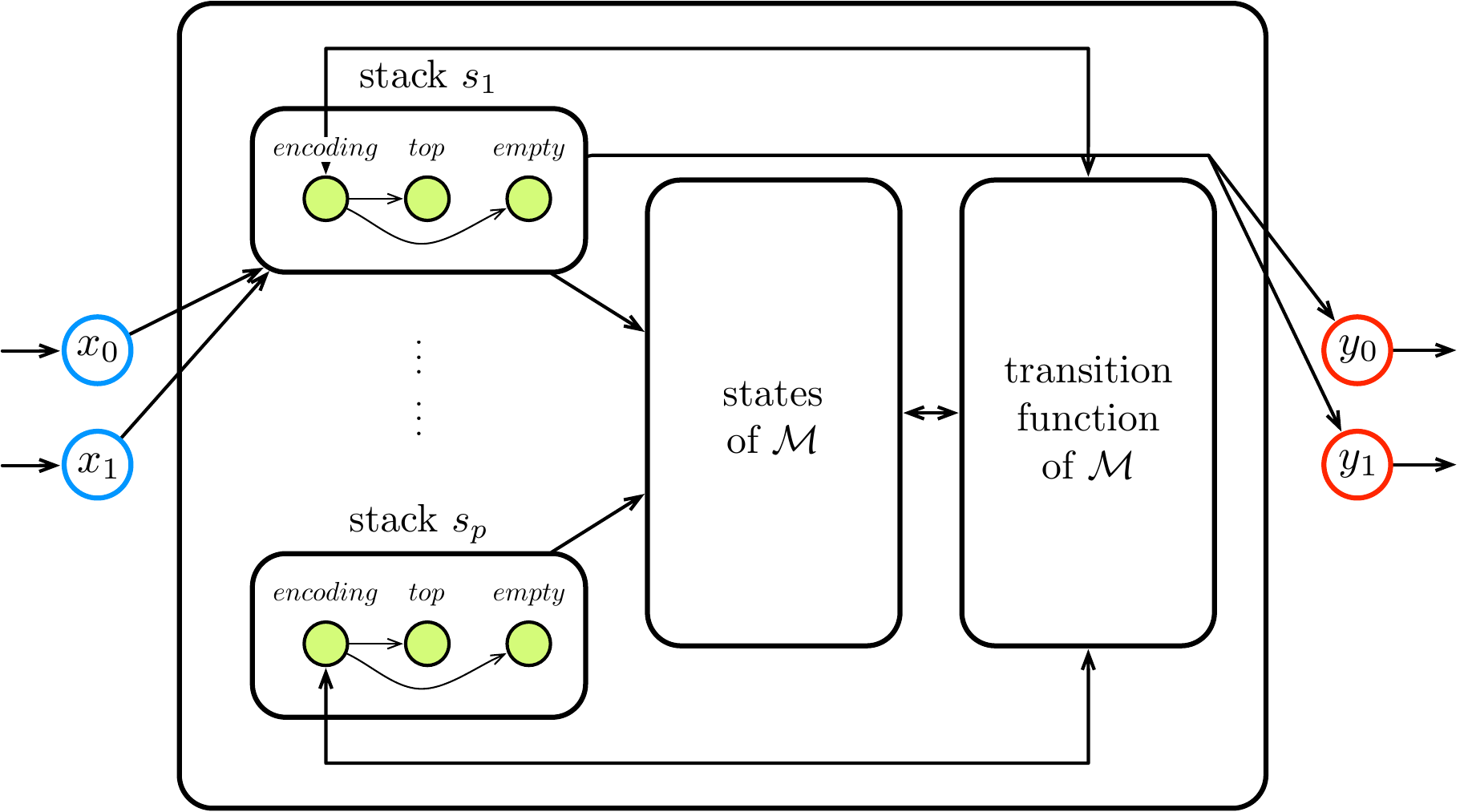}
\caption{Construction of an RNN that simulates a $p$-stack machine.}
\label{RNN_proof_fig}
\end{center}
\end{figure}

% PARAGRAPH %
It can be noticed that the simulation process described in the proof of Theorem~\ref{RRNN_thm1} is performed in real time. More precisely, if a language $L \subseteq \Sigma^*$ is decided by some TM in time $f(n)$, then $L$ is decided by some RNN in time $f(n) + O(n)$. Hence, when restricted to polynomial time of computation, RNNs decide the complexity class $\mathbf{P}$.

\begin{corollary}
\label{RRNN_thm2}
Let $L \subseteq \Sigma^*$ be some language. The following conditions are equivalent:
\begin{enumerate}[label=(\roman*),itemsep=0pt]
    \item $L \in \mathbf{P}$; \label{RRNN_thm2_c1}
    %\item $L$ is decidable by some TM in polynomial time; \label{RRNN_thm2_c2}
    \item $L$ is decidable by some RNN in polynomial time. \label{RRNN_thm2_c3}
\end{enumerate}
\end{corollary}

\section{Analog, Evolving and Stochastic Recurrent Neural Networks}
\label{sec:AES-NNs}

% PARAGRAPH %
We now introduce analog, evolving and stochastic recurrent neural networks, which are all variants of the RNN model. In polynomial time, these models capture the complexity classes $\mathbf{P/poly}$, $\mathbf{P/poly}$ and $\mathbf{BPP/log^*}$, respectively, which all strictly contain the class $\mathbf{P}$ and include non-recursive languages. According to these considerations, these augmented models have been qualified as {\it super-Turing}. For each model, a tight characterization in terms of Turing machines with specific kinds of advice is provided.

\subsection{Analog networks}

% PARAGRAPH %
An {\it analog recurrent neural network (ANN)} is an RNN as defined in Definition~\ref{RNN:def}, except that the weight matrices are real instead of rational~\cite{SiegelmannSontag94}. Formally, an ANN is an RNN
$$
\mathcal{N} = \left( \vec{x}, \vec{h}, \vec{y}, \vec{W_{in}},\vec{W_{res}}, \vec{W_{out}}, \vec{h}^0 \right)
$$
such that
$$
\vec{W_{in}} \in \mathbb{R}^{K \times (2+1)} \text{,~~} \vec{W_{res}} \in \mathbb{R}^{K \times K} \text{~~and~~} \vec{W_{out}} \in \mathbb{R}^{2 \times K}.
$$
The definitions of acceptance and rejection of words as well as of decision of languages is the same as for RNNs.

% PARAGRAPH %
It can been shown that any ANN $\mathcal{N}$ containing the irrational weights $r_1, \dots, r_k \in \mathbb{R} \setminus \mathbb{Q}$ is computationally equivalent to some ANN $\mathcal{N}'$ using only a single irrational weight $r \in \mathbb{R} \setminus \mathbb{Q}$ such that $r \in \Delta \subseteq [0, 1]$ and $r$ is the bias $w_{02}$ of the hidden cell $h_0$~\cite{SiegelmannSontag94}. 
%Indeed, by letting the successive bits of $r$ being defined as the interlacing of the successive bits of $r_1, \dots, r_n$, we can build a network ANN $\mathcal{N}'$ with weight $r$ which reconstructs the weights $r_1, \dots, r_n$ and simulates $\mathcal{N}$. Moreover, in this simulation process, we can assume that $r$ is in the interval $[0, 1]$ and is the bias of a hidden cell $h_i$, i.e., $r = w_{02} \in [0, 1]$ for some $0 \leq i \leq K-1$. 
Hence, without loss of generality, we restrict our attention to such networks. Let $r \in \Delta$ and $R \subseteq \Delta$.
\begin{itemize}[itemsep=0pt]
    \item $\text{ANN}[r]$ denotes the class of ANNs such that all weights but $w_{02}$ are rational and $w_{02} = r$.
    \item $\text{ANN}[R]$ denotes the class of ANNs such that all weights but $w_{02}$ are rational and $w_{02} \in R$.
\end{itemize}
In this definition, $r$ is allowed to be a rational number. In this case, an $\text{ANN}[r]$ is just a specific RNN.

% PARAGRAPH %
In exponential time of computation, analog recurrent neural networks can decide any possible language. In fact, any language $L \subseteq \Sigma^*$ can be encoded into the infinite word $\bar r_L \in \Sigma^\omega$, where the $i$-th bit of $\bar r_L$ equals $1$ iff the $i$-th word of $\Sigma^*$ belongs to $L$, according to some enumeration of $\Sigma^*$. Hence, we can build some ANN containing the real weight $r_L = \delta_4(\bar r_L)$, which, for every input $w$, decides whether $w \in L$ or $w \not \in L$ by decoding $\bar r_L$ and reading the suitable bit. In polynomial time of computation, however, the ANNs decide the complexity class $\mathbf{P/poly}$, and hence, are computationally equivalent to Turing machines with polynomial advice (TM/poly(A)). The following result holds~\cite{SiegelmannSontag94}:

\begin{theorem}
\label{ANN_thm}
Let $L \subseteq \Sigma^*$ be some language. The following conditions are equivalent:
\begin{enumerate}[label=(\roman*),itemsep=0pt]
    \item $L \in \mathbf{P/poly}$; \label{ANN_thm_c1}
    % \item $L$ is decidable by some TM/poly(A) in polynomial time; \label{ANN_thm_c2}
    \item $L$ is decidable by some $\text{ANN}$ in polynomial time. \label{ANN_thm_c3}
\end{enumerate}
\end{theorem}

% PARAGRAPH %
Given some $\text{ANN}$ $\mathcal{N}$ and some $q \in \mathbb{N}$, the {\it truncated network} $\mathcal{N} |_q$ is defined as the network $\mathcal{N}$ whose all weights and activation values are truncated after $q$ precision bits at each step of the computation. The following result shows that, up to time $q$, the network $\mathbb{N}$ can limit itself to $O(q)$ precision bits without affecting the result of its computation~\cite{SiegelmannSontag94}.

\begin{lemma}
\label{lemma_truncation_analog}
Let $\mathcal{N}$ be some $\text{ANN}$ computing in time $f : \mathbb{N} \rightarrow \mathbb{N}$. Then, there exists some constant $c > 0$ such that, for every $n \in \mathbb{N}$ and every input $w \in \Sigma^n$, the networks $\mathcal{N}$ and $\mathcal{N} |_{c f(n)}$ produce the same outputs up to time $f(n)$.
\end{lemma}

% PARAGRAPH %
The computational relationship between analog neural networks and Turing machines with advice can actually be strengthened. Towards this purpose, for any non-decreasing function $f : \mathbb{N} \rightarrow \mathbb{N}$ and any class of such functions $\mathcal{F}$, we define the following classes of languages decided by analog neural networks in time $f$ and $\mathcal{F}$, respectively:
\begin{eqnarray*}
\mathbf{ANN} \left[ r, f \right] & = & \left\{ L \subseteq \Sigma^\omega : L = L_f(\mathcal{N}) \text{ for some } \mathcal{N} \in \text{ANN}[r] \right\} \\
%\mathbf{ANN} \left[ r, \mathcal{F} \right] & = & \bigcup_{f \in \mathcal{F}} \mathbf{ANN} \left[r, f \right] \\
%\mathbf{ANN} \left[ R, f \right] & = & \bigcup_{r \in R} \mathbf{ANN} \left[ r, f \right] \\
\mathbf{ANN} \left[ R, \mathcal{F} \right] & = & \bigcup_{r \in R} \bigcup_{f \in \mathcal{F}} \mathbf{ANN} \left[ r, f \right] .
\end{eqnarray*}

% PARAGRAPH %
In addition, for any real $r \in \Delta$ and any function $f : \mathbb{N} \rightarrow \mathbb{N}$, the prefix advice $\alpha(\bar r, f) : \mathbb{N} \rightarrow \Sigma^*$ of length $f$ associated with $r$ is defined by 
%which, for any $n \in \mathbb{N}$, returns the $f(n)$ first bits of the binary expansion $\bar r$ of $r$, i.e.,
$$
\alpha(\bar r,f)(n) = r_0  r_1 \cdots r_{f(n)-1}
$$
for all $n \in \mathbb{N}$. For any set of reals $R \subseteq \Delta$ and any class of functions $\mathcal{F} \subseteq \mathbb{N}^\mathbb{N}$, we naturally set
$$
\alpha(\bar R, \mathcal{F}) = \bigcup_{\bar r \in \bar R} \bigcup_{f \in \mathcal{F}} \left\{ \alpha(\bar r,f) \right\}
$$
Conversely, note that any prefix unbounded advice $\alpha : \mathbb{N} \rightarrow \Sigma^*$ is of the form $\alpha(\bar r, f)$, where $\bar r = \lim_{n \rightarrow \infty} \alpha(n)  \in \Sigma^\omega$ and $f : \mathbb{N} \rightarrow \mathbb{N}$ is defined by $f(n) = | \alpha(n) |$.
%, we have
%$$
%\alpha(n) = \bar r_0 \bar r_1 \cdots \bar r_{f(n)-1} = \alpha(\bar r, f(n)
%$$
%for all $n \in \mathbb{N}$, and thus $\alpha = \alpha(\bar r, f)$.
%Hence, we define the real $r_\alpha \in \Delta$ naturally associated with $\alpha$ by
%$$
%r_\alpha = \delta_4( \lim_{n \rightarrow \infty} \alpha(n) ) .
%$$

% PARAGRAPH %
%The following result states that the computation times of analog neural networks and Turing machines with advice are polynomially related in the case where the real weight of the former being tightly linked to the advice of latter.
The following result clarifies the tight relationship between analog neural networks using real weights and Turing machines using related advices. Note that the real weights of the networks correspond precisely to the advice of the machines, and the computation time of the networks are related to the advice length of the machines.

\begin{proposition}
\label{ANN_prop}
Let $r \in \Delta$ be some real weight and $f : \mathbb{N} \rightarrow \mathbb{N}$ be some non-decreasing function. 
\begin{enumerate}[label=(\roman*),itemsep=0pt]
\item \label{ANN_prop_c1}
$\mathbf{ANN} \left[ r, f \right] \subseteq \mathbf{TMA} \left[ \alpha(\bar r, cf), O(f^3) \right]$, for some $c > 0$.
\item \label{ANN_prop_c2}
$\mathbf{TMA} \left[ \alpha(\bar r, f), f \right] \subseteq \mathbf{ANN} \left[ r,O(f) \right]$.
\end{enumerate}
\end{proposition}

\begin{proof}
\ref{ANN_prop_c1} 
Let $L \in \mathbf{ANN} \left[ r, f \right]$. Then, there exists some $\text{ANN}[r]$ $\mathcal{N}$ such that $L_f(\mathcal{N}) = L$. By Lemma~\ref{lemma_truncation_analog}, there exists some constant $c > 0$ such that the network $\mathcal{N}$ and the truncated network $\mathcal{N} |_{c f(n)}$ produce the same outputs up to time step $f(n)$, for all $n \in \mathbb{N}$. 
%Since $\mathcal{N}$ contains at most one irrational weight $r$, the result also holds true if the truncated network $\mathcal{N} |_{c f(n)}$ consists of $\mathcal{N}$ where only the weight $r$ is approximated by its $cf(n)$ first bits. 
Now, consider Procedure~\ref{algo1} below. In this procedure, all instructions except the query one (line 2) are recursive. Besides, the simulation of each step of $\mathcal{N} |_{c f(n)}$ involves a constant number of multiplications and additions of rational numbers, all representable by $c f(n)$ bits, and can thus be performed in time $O(f^2(n))$ (for the products). Consequently, the simulation of the $f(n)$ steps of $\mathcal{N} |_{c f(n)}$ can be performed in time $O(f^3(n))$. Hence, Procedure~\ref{algo1} can be simulated by some TM/A $\mathcal{M}$ using advice $\alpha(\bar r, cf)$ in time $O(f^3(n))$. In addition, Lemma~\ref{lemma_truncation_analog} ensures that $w$ is accepted by $\mathcal{M}$ iff $w$ is accepted by $\mathcal{N}$, for all $w \in \Sigma^*$. Hence, $L(\mathcal{M}) = L(\mathcal{N}) = L$, and therefore $L \in \mathbf{TMA} \left[ \alpha(\bar r, cf), O(f^3) \right]$.

\begin{algorithm}[h!]
\DontPrintSemicolon
\SetKwInOut{Input}{Inputs}
\smallskip
\KwInput{input $w \in \Sigma^n$}
{Query the advice $\alpha(\bar r, cf)(n) = r_0 r_1 \cdots r_{cf(n)-1}$}\;
%{Compute the rational approximation $\tilde r = \delta_4( r_0 r_1 \cdots r_{cf(n)-1} )$ of $r$}\;
\For{$t = 0, 1, \dots, f(n) - 1$}{
    {Simulate the truncated network $\mathcal{N} |_{c f(n)}$ which uses the rational approximation $\tilde r = \delta_4( r_0 r_1 \cdots r_{cf(n)-1} )$ of $r$ as its weight};
}
\Return Output of $\mathcal{N} |_{c f(n)}$ over $w$ at time step $f(n)$
\caption{}
\label{algo1}
\end{algorithm}

\ref{ANN_prop_c2} 
Let $L \in \mathbf{TMA} \left[ \alpha(\bar r, f), f \right]$. Then, there exists some TM/A $\mathcal{M}$ with advice $\alpha(\bar r, f)$ such that $L_f(\mathcal{M}) = L$. We show that $\mathcal{M}$ can be simulated by some analog neural network $\mathcal{N}$ with real weight $r$. The network $\mathcal{N}$ simulates the advice tape of $\mathcal{M}$ as described in the Proof of Theorem~\ref{RRNN_thm1}: the left and right contents of the tape are encoded and stored into two stack neurons $x_l$ and $x_r$, respectively, and the tape operations are simulated using appropriate neural circuits. On every input $w \in \Sigma^n$, the network $\mathcal{N}$ works as follows. First, $\mathcal{N}$ copies its real bias $r = \delta_4(r_0 r_1 r_2 \cdots)$ into some neuron $x_a$. Every time $\mathcal{M}$ reads some new advice bit $r_i$, then $\mathcal{N}$ first pops $r_i$ from its neuron $x_a$, which thus takes the updated activation value $\delta_4(r_{i+1} r_{i+2} \cdots)$, and next pushes $r_i$ into neuron $x_r$. This process is performed in constant time. At this point, neurons $x_l$ and $x_r$ contain the encoding of the bits $r_0 r_1 \cdots r_i$. Hence, $\mathcal{N}$ can simulates the recursive instructions of $\mathcal{M}$ in the usual way, in real time, until the next bit $r_{i+1}$ is read~\cite{SiegelmannSontag95}. Overall, $\mathcal{N}$ simulates the behavior of $\mathcal{M}$ in time $O(f(n))$. 

We now show that $\mathcal{M}$ and $\mathcal{N}$ output the same decision for input $w$. If $\mathcal{M}$ does not reach the end of its advice word $\alpha(n)$, the behaviors of $\mathcal{M}$ and $\mathcal{N}$ are identical, and so are their outputs. If at some point, $\mathcal{M}$ reaches the end of $\alpha(n)$ and reads successive blank symbols, then $\mathcal{N}$ continues to pop the successive bits $r_{|\alpha(n)|} r_{|\alpha(n)| + 1} \cdots$ from neuron $x_a$, to push them into neuron $x_r$, and to simulates the behavior of $\mathcal{M}$. In this case, $\mathcal{N}$ simulates the behavior of $\mathcal{M}$ working with some extension of the advice $\alpha(n)$, which by the consistency property of $\mathcal{M}$ (cf.~Section~\ref{sec:prelim}), produces the same output as if working with advice $\alpha(n)$. In this way, $w$ is accepted by $\mathcal{M}$ iff $w$ is accepted by $\mathcal{N}$, and thus $L(\mathcal{M}) = L(\mathcal{N})$. Therefore, $L \in \mathbf{ANN} \left[ r, O(f) \right]$.
\end{proof}

% PARAGRAPH %
The following corollary shows that the class of languages decided in polynomial time by analog networks with real weights and by Turing machines with related advices are the same.

\begin{corollary}
\label{ANN_cor}
Let $r \in \Delta$ be some real weight and $R \subseteq \Delta$ be some set of real weights. 
\begin{enumerate}[label=(\roman*),itemsep=0pt]
\item \label{ANN_cor_c1}
$\mathbf{ANN} \left[ r, \mathrm{poly} \right] = \mathbf{TMA} \left[ \alpha(\bar r,\mathrm{poly}), \mathrm{poly} \right]$.
\item \label{ANN_cor_c2}
$\mathbf{ANN} \left[ R, \mathrm{poly} \right] = \mathbf{TMA} \left[ \alpha(\bar R,\mathrm{poly}), \mathrm{poly} \right]$.
\end{enumerate}
\end{corollary}

\begin{proof}
\ref{ANN_cor_c1} Let $L \in \mathbf{ANN} \left[ r, \mathrm{poly} \right]$. Then, there exists $f \in \mathrm{poly}$ such that $L \in \mathbf{ANN} \left[ r, f \right]$.  By Proposition~\ref{ANN_prop}--\ref{ANN_prop_c1}, $L \in \mathbf{TMA} \left[ \alpha(\bar r, cf), O(f^3) \right]$, for some $c > 0$. Thus $L \in \mathbf{TMA} \left[ \alpha(\bar r, \mathrm{poly}), \mathrm{poly} \right]$. Conversely, let $L \in \mathbf{TMA} \left[ \alpha(\bar r, \mathrm{poly}), \mathrm{poly} \right]$. Then, there exist $f, f' \in \mathrm{poly}$ such that $L \in \mathbf{TMA} \left[ \alpha(\bar r, f'), f \right]$. By the consistency property of the TM/A, we can assume without loss of generality that $f' = f$. By Proposition~\ref{ANN_prop}--\ref{ANN_prop_c2}, $L \in \mathbf{ANN} \left[ r, O(f) \right]$, and hence $L \in \mathbf{ANN} \left[ r, \mathrm{poly} \right]$.

\ref{ANN_cor_c2} This point follows directly from point~\ref{ANN_cor_c1} by taking the union over all $r \in R$.
\end{proof}

\subsection{Evolving networks}

% PARAGRAPH %
An {\it evolving recurrent neural network (ENN)} is an RNN where the weight matrices can evolve over time inside a bounded space instead of staying static~\cite{CabessaSiegelmann14}. Formally, an ENN is a tuple
$$
\mathcal{N} = \left( \vec{x}, \vec{h}, \vec{y}, \left( \vec{W^\mathnormal{t}_{in}} \right)_{t \in \mathbb{N}}, \left( \vec{W^\mathnormal{t}_{res}} \right)_{t \in \mathbb{N}}, \left( \vec{W^\mathnormal{t}_{out}} \right)_{t \in \mathbb{N}}, \vec{h}^0 \right)
$$
where $\vec{x}, \vec{h}, \vec{y}, \vec{h}^0$ are defined as in Definition~\ref{RNN:def}, and 
$$
\vec{W^\mathnormal{t}_{in}} \in \mathbb{Q}^{K \times (2+1)} \text{,~~} \vec{W^\mathnormal{t}_{res}} \in \mathbb{Q}^{K \times K} \text{~~and~~} \vec{W^\mathnormal{t}_{out}} \in \mathbb{Q}^{2 \times K}.
$$
are input, reservoir and output weight matrices at time $t$ such that $\| \vec{W^\mathnormal{t}_{in}} \|_{\max} = \| \vec{W^\mathnormal{t}_{res}} \|_{\max} = \| \vec{W^\mathnormal{t}_{out}} \|_{\max} \leq C$ for some constant $C > 1$ and for all $t \in \mathbb{N}$. The boundedness condition expresses the fact that the synaptic weights are confined into a certain range of values imposed by the biological constitution of the neurons. The successive values of an evolving weight $w_{ij}$ is denoted by $(w^t_{ij})_{t \in \mathbb{N}}$. The dynamics of an ENN is given by the following adapted equations
\begin{eqnarray}
    \vec{h}^{t+1} & = & \sigma \left( \vec{W^\mathnormal{t}_{in}} (\vec{x}^t : 1) + \vec{W^\mathrm{t}_{res}} \vec{h}^t \right) \label{ENN:eq1} \\ 
    \vec{y}^{t+1} & = & \theta \left( \vec{W^\mathnormal{t}_{out}} \vec{h}^{t+1} \right) \label{ENN:eq2} .
\end{eqnarray}
The definition of acceptance and rejection of words and decision of languages is the same as for RNNs.

% PARAGRAPH %
In this case also, it can been shown that any ENN $\mathcal{N}$ containing the evolving weights $\bar e_1, \dots, \bar e_n \in [-C, C]^\omega$ is computationally equivalent to some ENN $\mathcal{N}'$ containing only one evolving weight $\bar e \in [-C, C]^\omega$, such that $\bar e$ evolves only among the binary values $0$ and $1$, i.e. $\bar e \in \Sigma^\omega$, and $\bar e$ is the evolving bias $(w^t_{02})_{t \in \mathbb{N}}$ of the hidden cell $h_0$~\cite{CabessaSiegelmann14}. Hence, without loss of generality, we restrict our attention to such networks. Let $\bar e \in \Sigma^\omega$ be some binary evolving weight and $\bar E \subseteq \Sigma^\omega$.
\begin{itemize}[itemsep=0pt]
    \item $\text{ENN}[\bar e]$ denoted the class of ENNs such that all weights but $w_{02}$ are static, and $(w^t_{02})_{t \in \mathbb{N}} = \bar e$.
    \item $\text{ENN}[\bar E]$ denotes the class of ENNs such that all weights but $w_{02}$ are static, and $(w^t_{02})_{t \in \mathbb{N}} \in \bar E$.
\end{itemize}

% PARAGRAPH %
Like analog networks, evolving recurrent neural networks can decide any possible language in exponential time of computation. In polynomial time, they decide the complexity class $\mathbf{P/poly}$, and thus, are computationally equivalent to Turing machines with polynomial advice (TM/poly(A)). The following result holds~\cite{CabessaSiegelmann14}:

\begin{theorem}
\label{ENN_thm}
Let $L \subseteq \Sigma^*$ be some language. The following conditions are equivalent:
\begin{enumerate}[label=(\roman*),itemsep=0pt]
    \item $L \in \mathbf{P/poly}$; \label{ENN_thm_c1}
    % \item $L$ is decidable by some TM/poly(A) in polynomial time; \label{ENN_thm_c2}
    \item $L$ is decidable by some $\text{ENN}$ in polynomial time. \label{ENN_thm_c3}
\end{enumerate}
\end{theorem}

% PARAGRAPH %
An analogous version of Lemma~\ref{lemma_truncation_analog} holds for the case of evolving networks~\cite{CabessaSiegelmann14}. Note that the boundedness condition on the weights is involved in this result.

\begin{lemma}
\label{lemma_truncation_evolving}
Let $\mathcal{N}$ be some $\text{ENN}$ computing in time $f : \mathbb{N} \rightarrow \mathbb{N}$. Then, there exists some constant $c$ such that, for every $n \in \mathbb{N}$ and every input $w \in \Sigma^n$, the networks $\mathcal{N}$ and $\mathcal{N} |_{c f(n)}$ produce the same outputs up to time $f(n)$.
\end{lemma}

% PARAGRAPH %
Once again, the computational relationship between evolving neural networks and Turing machines with advice can be strengthen. For this purpose, we define the following classes of languages decided by evolving neural networks in time $f$ and $\mathcal{F}$, respectively:
\begin{eqnarray*}
\mathbf{ENN} \left[ \bar e, f \right] & = & \left\{ L \subseteq \Sigma^\omega : L = L_f(\mathcal{N}) \text{ for some } \mathcal{N} \in \text{ENN}[\bar e] \right\} \\
%\mathbf{ENN} \left[\bar e, \mathcal{F} \right] & = & \bigcup_{f \in \mathcal{F}} \mathbf{ENN} \left[\bar e, f \right] \\
%\mathbf{ENN} \left[\bar E, f \right] & = & \bigcup_{\bar e \in \bar E} \mathbf{ENN} \left[ \bar e, f \right] \\
\mathbf{ENN} \left[ \bar E, \mathcal{F} \right] & = & \bigcup_{\bar e \in \bar E} \bigcup_{f \in \mathcal{F}} \mathbf{ENN} \left[ \bar e, f \right] .
\end{eqnarray*}
% PARAGRAPH %
For any $\bar e \in \Sigma^\omega$ and any function $f : \mathbb{N} \rightarrow \mathbb{N}$, we consider the prefix advice $\alpha(\bar e, f) : \mathbb{N} \rightarrow \Sigma^*$ associated with $e$ and $f$ defined by
$$
\alpha(\bar e, f)(n) = e_{0}e_{1} \cdots e_{f(n)-1}
$$
for all $n \in \mathbb{N}$. Conversely, any prefix advice $\alpha : \mathbb{N} \rightarrow \Sigma^*$ is clearly of the form $\alpha(\bar e, f)$, where $\bar e = \lim_{n \rightarrow \infty} \alpha(n) \in \Sigma^\omega$ and $f(n) = | \alpha(n) |$ for all $n \in \mathbb{N}$.

% PARAGRAPH %
The following relationships between neural networks with evolving weights and Turing machines with related advice hold:

\begin{proposition}
\label{ENN_prop}
Let $e \in \Sigma^\omega$ be some binary evolving weight and $f : \mathbb{N} \rightarrow \mathbb{N}$ be some non-decreasing function. 
\begin{enumerate}[label=(\roman*),itemsep=0pt]
\item \label{ENN_prop_c1}
$\mathbf{ENN} \left[ \bar e, f \right] \subseteq \mathbf{TMA} \left[ \alpha(\bar e, c f), O(f^3) \right]$, for some $c > 0$.
\item \label{ENN_prop_c2}
$\mathbf{TMA} \left[ \alpha(\bar e, f), f \right] \subseteq \mathbf{ENN} \left[ \bar e, O(f) \right]$.
\end{enumerate}
\end{proposition}

\begin{proof}
The proof is very similar to that of Proposition~\ref{ANN_prop}.

\ref{ENN_prop_c1} 
Let $L \in \mathbf{ENN} \left[ \bar e, f \right]$. Then, there exists some $\text{ENN}[\bar e]$ $\mathcal{N}$ such that $L_f(\mathcal{N}) = L$. By Lemma~\ref{lemma_truncation_evolving}, there exists some constant $c > 0$ such that the network $\mathcal{N}$ and the truncated network $\mathcal{N} |_{c f(n)}$ produce the same outputs up to time step $f(n)$, for all $n \in \mathbb{N}$. Now, consider Procedure~\ref{algo2} below. In this procedure, all instructions except the query one are recursive. Procedure~\ref{algo2} can be simulated by some TM/A $\mathcal{M}$ using advice $\alpha(\bar e, f)$ in time $O(f^3(n))$, as described in the proof of Proposition~\ref{ANN_prop}. In addition, $\mathcal{M}$ and $\mathcal{N}$ decide the same language $L$, and therefore $L \in \mathbf{TMA} \left[ \alpha(\bar e, f), O(f^3) \right]$.

\begin{algorithm}[h!]
\DontPrintSemicolon
\SetKwInOut{Input}{Inputs}
\smallskip
\KwInput{input $w \in \Sigma^n$}
{Query the advice $\alpha(\bar e, f)(n) = e_0 e_1 \cdots e_{cf(n)-1}$}\;
\For{$t = 0, 1, \dots, f(n)$}{
    {Simulate the truncated network $\mathcal{N} |_{c f(n)}$ which uses $w^t_{02} = e_t$ as the current value if its evolving weight};
}
\Return Output of $\mathcal{N} |_{c f(n)}$ over $w$ at time step $f(n)$
\caption{}
\label{algo2}
\end{algorithm}

\ref{ANN_prop_c2} 
Let $L \in \mathbf{TMA} \left[ \alpha(\bar e, f), f \right]$. Then, there exists to some TM/A $\mathcal{M}$ with advice $\alpha(\bar e, f)$ such that $L_f(\mathcal{M}) = L$. The machine $\mathcal{M}$ can be simulated by the network $\mathcal{N}$ with evolving weight $\bar e = e_0 e_1 e_2 \cdots$ as follows. First, $\mathcal{N}$ simultaneously counts and pushes into a stack neuron $x_a$ the successive bits of $\bar e$ as they arrive. Then, for $k = 1, 2, \dots$ and until it produces a decision, $\mathcal{N}$ proceeds as follows. If necessary, $\mathcal{N}$ waits for $x_a$ to contain more than $2^k$ bits, copies the content $e_0 e_1 \cdots e_{2^k} \cdots$ of $x_a$ in reverse order into another stack neuron $x_{a'}$, and simulates $\mathcal{M}$ with advice $e_0 e_1 \cdots e_{2^k} \cdots$ in real time. Every time $\mathcal{M}$ reads a new advice bit, $\mathcal{N}$ tries to access it from its stack $x_{a'}$. If $x_{a'}$ does not contain this bit, then $\mathcal{N}$ restart the whole process with $k + 1$. When $k = \log(f(n))$, the stack $x_a$ contains $2^k = f(n)$ bits, which ensures that $\mathcal{N}$ properly simulates $\mathcal{M}$ with advice $e_0 e_1 \cdots e_{f(n)-1}$. Hence, the whole simulation process is achieved in time $O( \sum_{k=1}^{\log(f(n))} 2^k) = O(2^{\log(f(n)) + 1}) = O(f(n))$. 
%in a similar way as described in the proof of Proposition~\ref{ANN_prop}. 
%In this case however, $\mathcal{N}$ pushes the successive bits of $\bar e$ as they arrive into some neuron $x_a$. In addition, $\mathcal{N}$ simulates an additional counter that stores the maximal index $i$ of the bit $e_i$ that has been read so far by $\mathcal{M}$. Then, every time $\mathcal{M}$ reads some new advice bit $e_{i+1}$, $\mathcal{N}$ updates its counter to $i+1$, pops $e_{i+1}$ from its neuron $x_a$, and then simulates $\mathcal{M}$ as usual. 
In this way, $\mathcal{M}$ and $\mathcal{N}$ decide the same language $L$, and $\mathcal{M}$ is simulated by $\mathcal{N}$ in time $O(f)$. Therefore, $L \in \mathbf{ENN} \left[ \bar e, O(f) \right]$.
\end{proof}

% PARAGRAPH %
The class of languages decided in polynomial time by evolving networks and Turing machines using related evolving weights and advices are the same.

\begin{corollary}
\label{ENN_cor}
Let $\bar e \in \Sigma^\omega$ be some binary evolving weight and $\bar E \subseteq \Sigma^\omega$ be some set of binary evolving weights. 
\begin{enumerate}[label=(\roman*),itemsep=0pt]
\item \label{ENN_cor_c1}
$\mathbf{ENN} \left[ \bar e, \mathrm{poly} \right] = \mathbf{TMA} \left[ \alpha(\bar e, \mathrm{poly}), \mathrm{poly} \right]$.
\item \label{ENN_cor_c2}
$\mathbf{ENN} \left[ \bar E, \mathrm{poly} \right] = \mathbf{TMA} \left[ \alpha(\bar E, \mathrm{poly}), \mathrm{poly} \right]$.
\end{enumerate}
\end{corollary}

\begin{proof}
The proof is similar to that of Corollary~\ref{ANN_cor}.
\end{proof}

\subsection{Stochastic networks}

% PARAGRAPH %
A {\it stochastic recurrent neural network (SNN)} is an RNN as defined in Definition~\ref{RNN:def}, except that the network contains additional stochastic cells as inputs~\cite{Siegelmann99}. Formally, an SNN is an RNN
$$
\mathcal{N} = \left( \vec{x}, \vec{h}, \vec{y}, \vec{W_{in}},\vec{W_{res}}, \vec{W_{out}}, \vec{h}^0 \right)
$$
such that $\vec{x} = (x_0, x_1, x_2, \cdots, x_{k+1})$, where $x_0$ and $x_1$ are the data and validation input cells, respectively, and $x_2, \dots, x_{k+1}$ are $k$ additional stochastic cells. The dimension of the input weight matrix is adapted accordingly, namely $\vec{W_{in}} \in \mathbb{R}^{K \times ((k+2)+1)}$. Each stochastic cell $x_i$ is associated with a probability $p_i \in [0, 1]$, and at each time step $t \geq 0$, the activation of the cell $x^t_i$ takes value $1$ with probability $p_i$, and value $0$ with probability $1 - p_i$. The dynamics of an SNN is then governed by Equations~(\ref{rnn:eq1}) and (\ref{rnn:eq2}), but with the adapted inputs $\vec{x}^t = (x^t_0, x^t_1, x^t_2, \cdots, x^t_{k+1}) \in \mathbb{B}^{k+2}$ for all $t \geq 0$.

% PARAGRAPH %
For some SNN $\mathcal{N}$, we assume that any input $w = w_0 w_1 \cdots w_{n-1} \in \Sigma^*$ is decided by $\mathcal{N}$ in the same amount of time $\tau(n)$, regardless of the random pattern of the stochastic cells $x^t_i \in \{ 0, 1 \}$, for all $i = 2, \dots, k+1$. Hence, the number of possible computations of $\mathcal{N}$ over $w$ is finite. The input $w$ is {\it accepted} (resp.~{\it rejected}) by $\mathcal{N}$ if the number of accepting (resp.~rejecting) computations over the total number of computations on $w$ is greater than or equal to $2/3$. 
%The input $w$ is {\it rejected} by $\mathcal{N}$ if the number of rejecting outcomes over the total number of outcomes is also greater than or equal to $2/3$. 
This means that the error probability of $\mathcal{N}$ is bounded by $1/3$. If $f : \mathbb{N} \rightarrow \mathbb{N}$ is a non-decreasing function, we say that $w$ is {\it accepted} or {\it rejected} by $\mathcal{N}$ {\it in time} $f$ if it is accepted or rejected in time $\tau(n) \leq f(n)$, respectively. We assume that any SNN is a decider. The definition of decision of languages is the same as in the case of RNNs.

% PARAGRAPH %
Once again, any SNN is computationally equivalent to some SNN with only one stochastic cell $x_2$ associated with a real probability $p \in \Delta$~\cite{Siegelmann99}. Without loss of generality, we restrict our attention to such networks. Let $p \in \Delta$ be some probability and $P \subseteq \Delta$.
\begin{itemize}[itemsep=0pt]
    \item $\text{SNN}[p]$ denotes the class of SNNs such that the probability of the stochastic cell $x_2$ is equal to $p$.
    \item $\text{SNN}[P]$ denotes the class of SNNs such that the probability of the stochastic cell $x_2$ is equal to some $p \in P$.
\end{itemize}

% PARAGRAPH %
In polynomial time of computation, the SNNs with rational probabilities decide the complexity class $\mathbf{BPP}$. By contrast, the SNNs with real probabilities decide the complexity class $\mathbf{BPP/log^*}$, and hence, are computationally equivalent to probabilistic Turing machines with logarithmic advice (PTM/log(A)). The following result holds~\cite{Siegelmann99}:

\begin{theorem}
\label{SNN_thm}
Let $L \subseteq \Sigma^*$ be some language. The following conditions are equivalent:
\begin{enumerate}[label=(\roman*),itemsep=0pt]
    \item $L \in \mathbf{BPP/log^*}$; \label{SNN_thm_c1}
    % \item $L$ is decidable by some PTM/log(A) in polynomial time; \label{SNN_thm_c2}
    \item $L$ is decidable by some $\text{SNN}$ in polynomial time. \label{SNN_thm_c3}
\end{enumerate}
\end{theorem}

% PARAGRAPH %
As for the two previous models, the computational relationship between stochastic neural networks and Turing machines with advice can be precised. We define the following classes of languages decided by stochastic neural networks in time $f$ and $\mathcal{F}$, respectively:
\begin{eqnarray*}
\mathbf{SNN}\left[ p, f \right] & = & \left\{ L \subseteq \Sigma^\omega : L = L_f(\mathcal{N}) \text{ for some } \mathcal{N} \in \text{SNN}[p] \right\} \\
%\mathbf{SNN} \left[ p, \mathcal{F} \right] & = & \bigcup_{f \in \mathcal{F}} \mathbf{SNN} \left[ p, f \right] \\
%\mathbf{SNN} \left[ P, f \right] & = & \bigcup_{p \in P} \mathbf{SNN} \left[ p, f \right] \\
\mathbf{SNN} \left[ P, \mathcal{F} \right] & = & \bigcup_{p \in P}  \bigcup_{f \in \mathcal{F}} \mathbf{SNN} \left[ p, f \right] .
\end{eqnarray*}

% PARAGRAPH %
%The computation times of stochastic neural networks and probabilistic Turing machines with advice using related stochasticity and advice, respectively, are polynomially related. But in this context, the length of the Turing machine's advice is logarithmically related to the computation time of the stochastic network. 
The tight relationships between stochastic neural networks using real probabilities and Turing machines using related advices can now be described. In this case however, the advice of the machines are logarithmically related to the computation time of the networks.

\begin{proposition}
\label{SNN_prop}
Let $p \in \Delta$ be some real probability and $f : \mathbb{N} \rightarrow \mathbb{N}$ be some non-decreasing function. 
\begin{enumerate}[label=(\roman*),itemsep=0pt]
\item \label{SNN_prop_c1}
%$\mathbf{SNN} \left[ p, f \right] \subseteq \mathbf{PTMA}[\alpha(\bar p, \log(f)+3), O(f\log(f))]$. % Yann's
$\mathbf{SNN}[p, f] \subseteq \mathbf{PTMA} \left[ \alpha(\bar p, \log(5f)), O(f^3) \right]$. 
\item \label{SNN_prop_c2}
$\mathbf{PTMA} \left[ \alpha(\bar p, \log(f)), f \right] \subseteq \mathbf{SNN} \left[ p,O(f^2) \right]$.
\end{enumerate}
\end{proposition}

\begin{proof}

\ref{SNN_prop_c1} Let $L \in \mathbf{SNN} \left[ p, f \right]$. Then, there exists some $\text{SNN}[p]$ $\mathcal{N}$ deciding $L$ in time $f$. Note that the stochastic network $\mathcal{N}$ can be considered as a classical rational-weighted RNN with an additional input cell $x_2$. Since $\mathcal{N}$ has rational weights, it can be noticed that up to time $f(n)$, the activation values of its neurons are always representable by rationals of $O(f(n))$ bits. 
%Hence, by Lemma~\ref{lemma_truncation_analog}, there exists some constant $c > 0$ such that the network $\mathcal{N}$ and the truncated network $\mathcal{N} |_{c f(n)}$ produce the same outputs up to time step $f(n)$, for all $n \in \mathbb{N}$.\footnote{Since Lemma~\ref{lemma_truncation_analog} holds for ANNs, it clearly also holds for RNNs, given that latter networks are particular cases of the former.} Without loss of generality, we may assume that $c > 5$. 
Now, consider Procedure~\ref{algo3} below. This procedure can then be simulated by some PTM/A $\mathcal{M}$ using advice $\alpha(\bar p, \log(5f))$ in time $O(f^3)$, as described in the proof of Proposition~\ref{ANN_prop}. 

It remains to show that $\mathcal{N}$ and $\mathcal{M}$ decide the same language $L$. For this purpose, consider a hypothetical device $\mathcal{M}'$ working as follows: at each time $t$, $\mathcal{M}'$ takes the sequence of bits $\bar b$ generated by Procedure~\ref{algo3} and concatenates it with some infinite sequence of bits $\bar b' \in \Sigma^\omega$ drawn independently with probability $\tfrac{1}{2}$, thus producing the infinite sequence $\bar b \bar b' \in \Sigma^\omega$. Then, $\mathcal{M}'$ generates a bit $c'_t$ iff $\bar b \bar b' <_{lex} \bar p$, which happens precisely with probability $p$, since $p = \delta_2(\bar p)$~\cite{AroraBarak06}. Finally, $\mathcal{M}'$ simulates the behavior of $\mathcal{N}$ at time $t$ using the stochastic bit $x_2^t = c'_t$. Clearly, $\mathcal{M}'$ and $\mathcal{N}$ produce random bits with same probability $p$, behave in the same way, and thus decide the same language $L$. We now evaluate the error probability of $\mathcal{M}$ at deciding $L$, by comparing the behaviors of $\mathcal{M}$ and $\mathcal{M}'$. Let $w \in \Sigma^n$ be some input and let
$$
\bar p_{\mathcal{M}} = \alpha(\bar p, \log(5f))(n) = p_0 p_1 \cdots p_{\log(5f(n))-1} \text{ ~and~ } p_{\mathcal{M}} = \delta_2(\bar p_{\mathcal{M}}) .
$$
According to Procedure~\ref{algo3}, at each time step $t$, the machine $\mathcal{M}$ generates $c_t = 1$ iff $\bar b <_{lex} \bar p_{\mathcal{M}}$, which happens precisely with probability $p_{\mathcal{M}}$, since $p_{\mathcal{M}} = \delta_2(\bar p_{\mathcal{M}})$~\cite{AroraBarak06}. On the other hand, $\mathcal{M}'$ generates $c'_t = 1$, with probability $p$, showing that $\mathcal{M}$ and $\mathcal{M}'$ might differ in their decisions. 
Since $\bar p_{\mathcal{M}}$ is a prefix of $\bar p$, it follows that $p_{\mathcal{M}} \leq p$ and
$$
p - p_{\mathcal{M}} = \sum_{i=\log(5f(n))}^{\infty} \frac{p_{i}}{2^{i+1}} \leq \frac{1}{2^{\log(5f(n))}} = \frac{1}{5f(n)} .
$$
In addition, the bits $c_t$ and $c'_t$ are generated by $\mathcal{M}$ and $\mathcal{M}'$ based on the sequences $\bar b$ and $\bar b \bar b'$ satisfying $\bar b <_{lex} \bar b \bar b'$. Hence,
$$
\mathrm{Pr}(c_t \neq c'_t) = \mathrm{Pr}(\bar p_{\mathcal{M}} <_{lex} \bar b \bar b' <_{lex} \bar p) = p - p_{\mathcal{M}} \leq \frac{1}{5f(n)} .
$$
By a union bound argument, the probability that the sequences $\bar c = c_0 c_1 \cdots c_{f(n)-1}$ and $\bar c' = c'_0 c'_1 \cdots c'_{f(n)-1}$ generated by $\mathcal{M}$ and $\mathcal{M}'$ differ satisfies
$$
\mathrm{Pr}(\bar c \neq \bar c') \leq \frac{1}{5f(n)} f(n) = \frac{1}{5} \text{ ~and thus~ } \mathrm{Pr}(\bar c = \bar c') \geq 1 - \frac{1}{5} .
$$
Since $\mathcal{M}'$ classifies $w$ correctly with probability at least $\tfrac{2}{3}$, it follows that $\mathcal{M}$ classifies $w$ correctly with probability at least $(1 - \frac{1}{5}) \tfrac{2}{3} > \tfrac{8}{15} > \tfrac{1}{2}$. 
%As a consequence, $\mathcal{M}$ and $\mathcal{M}'$ agree on their decisions about $w$ with probability at least $1 - \tfrac{1}{c}$. 
%When $w \in L$ (resp. $w \notin L$), it is accepted (resp. rejected) with probability at least $2/3$ by $\mathcal{N}$ and thus it is accepted (resp.rejected) by $\mathcal{M}$ with probability at least $2/3( 1 - 1/c)$. If we let $c=5$,  $2/3(1 - 1/5) > 1/2$. 
This probability can be increased above $\tfrac{2}{3}$ by repeating Procedure~\ref{algo3} a constant number of time and taking the majority of the decisions as output~\cite{AroraBarak06}. Consequently, the devices $\mathcal{M}$, $\mathcal{M}'$ and $\mathcal{N}$ all decide the same language $L$, and therefore, $L \in \mathbf{PTMA} \left[ \alpha(\bar p, \log(5f)), O(f^3) \right]$.

\begin{algorithm}[h!]
\DontPrintSemicolon
\SetKwInOut{Input}{Inputs}
\smallskip
\KwInput{input $w \in \Sigma^n$}
{Query the advice $\alpha(\bar p, \log(cf))(n) = \bar p_{\mathcal{M}} = p_0 p_1 \cdots p_{\log(5f(n))-1}$}\;
\For{$t = 0, 1, \dots, f(n) - 1$}{
    %\For{$i = 0, 1, \dots, \log(cf(n)) - 1$}{
    %    %{Draw the bit $x_t$ with probability $p_{\mathcal{M}}$ \label{line:draw}}\;
    %    {Draw a bit $b_i$ with probability $\tfrac{1}{2}$}\;
    %}
    {Draw a sequence of independent fair bits $\bar b := b_0 b_1 \cdots b_{\log(5f(n))-1}$}\;
    %{$\bar b := b_0 b_1 \cdots b_{\log(cf(n))-1}$}\;
    \lIf{$\bar b <_{lex} \bar p_{\mathcal{M}}$}{
    $c_t = 1$
    }
    \lElse{
    $c_t = 0$
    }
    {Simulate network $\mathcal{N} |_{c f(n)}$ at time $t$ with $x_2^t = c_t$}\;
}
\Return Output of $\mathcal{N} |_{c f(n)}$ over $w$ at time step $f(n)$
\caption{}
\label{algo3}
\end{algorithm}

%% Yann's proof
%(ii) Let now $M \in \mathbf{PTMA}[\alpha(\bar p, \log(f)),f]$, we give the construction of $\mathcal{N} \in \mathbf{SNN}[p, f]$ which decides the same language. On an input $x$, the network $\mathcal{N}$ evaluates the first $\log(f(n))$ bits of $p$.
%To do that, it draws $k$ bits which are $1$ with probability $p$ and compute their sum in time $k$. Then, it divides this sum by $k$ in time $O(k)$ and keep the first $\log(f(n))$ bits. Finally, it uses these bits as an advice and simulates the machine $M$ as in~\cite{Siegelmann95}. When the bits are equal to $\alpha(\bar p, \log(f))$, the simulation is correct. We thus need to evaluate the probability that the equality holds. By Chebyshev's inequality, the probability that the difference from our estimated $p$ and the real $p$ is larger than $1/2f(n)$ is less than $\frac{p(1-p)4f(n)^2}{k}$. If we let $k = 8(p(1-p)4f(n)^2)$,
%the probability of error is less than $1/8$. Hence, the probability that $\mathcal{N}$ simulates $M$ is strictly larger than $1/2$ and as before, it can be increased to $2/3$ by repeating the simulation a constant number of time. The simulation of $M$ is in real time and the only overhead is the evaluation of $p$ in time $O(f(n)^2)$.

\ref{SNN_prop_c2} Let $L \in \mathbf{PTMA}[\alpha(\bar p, \log(f)), f]$. Then, there exists some PTM/log(A) $\mathcal{M}$ with logarithmic advice $\alpha(\bar p, \log(f))$ deciding $L$ in time $f$. 
%c By definition, $\mathcal{M}$ classifies every word $w \in \Sigma^*$ with some error probability bounded by $\tfrac{1}{3}$. 
For simplicity purposes, let the advice of $\mathcal{M}$ be denoted by $\bar p = p_0 \cdots p_{\log(f(n))-1}$ ($\bar p$ is not anymore the binary expansion of $p$ from now on). 
% By repeating the computation of $\mathcal{M}$ a constant number of times and taking as decision the majority of these outcomes, we can assume without loss of generality that $\mathcal{M}$ classifies every word $w \in \Sigma^*$ with some error probability $\epsilon$ satisfying
% $$
% \epsilon < \frac{1}{6e^{-1/\log(\tilde p)}}
% $$
% where $\tilde p = p^2 (1 - p)^2$ (for instance, if $p = 0.7$, then $\epsilon \simeq 0.0024$ )~\cite{AroraBarak06}. 
Now, consider Procedure~\ref{algo4} below. The first for loop computes an estimation $\bar p'$ of the advice $\bar p$ defined by
$$
\bar p' = \delta_2^{-1} ( p' ) [0:\log(f(n))-1] = p'_0 \cdots p'_{\log(f(n))-1} 
$$
where 
$$
p' = \frac{1}{k(n)} \sum_{i=0}^{k(n)-1} b_i \text{ ~and~ } k(n) = \ulcorner 10 p (1-p) f^2(n) \urcorner \,
$$
and the $b_i$ are drawn according to a Bernouilli distribution of parameter $p$.
The second for loop computes a sequence of random choices
$$
\bar c = c_0 \cdots c_{f(n)-1}
$$
using von Neumann's trick to simulate a fair coin with a biased one~\cite{AroraBarak06}. The third loop simulates the behavior of the PTM/log(A) $\mathcal{M}$ using the alternative advice $\bar p'$ and the sequence of random choices $\bar c$. This procedure can clearly be simulated by some $\text{SNN}[p]$ $\mathcal{N}$ in time $O(k + 2f) = O(f^2)$, where the random samples of bits are given by the stochastic cell and the remaining recursive instructions simulated by a rational-weighted sub-network. 

It remains to show that $\mathcal{M}$ and $\mathcal{N}$ decide the same language $L$. For this purpose, we estimate the error probability of $\mathcal{N}$ at deciding language $L$. First, we show that $\bar p'$ is a good approximation of the advice $\bar p$ of $\mathcal{M}$.
%denoted for simplicity as
%$$
%\bar p = \alpha(\bar p, \log(f))(n) = p_0 \cdots p_{\log(f(n))-1} .
%$$
Note that $\bar p' \neq \bar p$ iff $|p' - p| > \tfrac{1}{2^{\log (f(n))}} = \tfrac{1}{f(n)}$. 
%Hence, $\mathrm{Pr}(\bar p '\neq \bar p) = \mathrm{Pr}(|p' - p| > \tfrac{1}{f(n)})$. 
Note also that by definition, $p' = \tfrac{\#1}{k(n)}$, where $\#1 \sim \mathcal{B}(k(n), p)$ is a binomial random variable of parameters $k(n)$ and $p$ with $\mathrm{E}(\#1) = k(n)p$ and $\mathrm{Var}(\#1) = k(n) p (1-p)$. It follows that
\begin{eqnarray*}
    \mathrm{Pr} \left( \bar p' \neq \bar p \right) & = & \mathrm{Pr}\left( |p' - p| > \frac{1}{f(n)} \right) \\
    & = & \mathrm{Pr}\left( |k(n)p' - k(n)p| > \frac{k(n)}{f(n)} \right) \\
    & = & \mathrm{Pr}\left( |\#1 - \mathrm{E}(\#1)| > \frac{k(n)}{f(n)} \right) .
\end{eqnarray*}
The Chebyshev's inequality ensures that
$$
\mathrm{Pr} \left( \bar p' \neq \bar p \right)  \leq  \frac{\mathrm{Var}(\#1) f^2(n)}{k^2(n)} = \frac{p (1-p) f^2(n)}{k(n)} < \frac{1}{10}
$$
since $k(n) > 10 p (1-p) f^2(n)$. 

We now estimate the source of error coming from the simulation of a fair coin by a biased one in Procedure~\ref{algo4} (loop of Line~\ref{trick}). Note that at each step $i$, if the two bits $bb'$ are different ($01$ or $10$), then $c_t$ is drawn with fair probability $\tfrac{1}{2}$, like in the case of the machine $\mathcal{M}$. Hence, the sampling process of $\mathcal{N}$ and $\mathcal{M}$ differ in probability precisely when all of the $K$ draws produce identical bits $bb'$ ($00$ or $11$). The probability that the two bits $bb'$ are identical at step $i$ is $p^2 + (1-p)^2$, and hence, the probability that the $K = \frac{- 4 - \log(f(n))}{\log(p^2+(1-p)^2)}$ independent draws all produce identical bits $bb'$ satisfies
$$
\left( p^2 + (1-p)^2 \right)^{K} \leq 2^{- 4 - \log(f(n))} = \frac{1}{16f(n)}.
$$
by using the fact that $x^{1/\log(x)} \geq 2$.
%By a union bound argument, the probability that some $c_t$ in the sequence $c_0 \cdots c_{f(n)-1}$, was not computed from drawing two different bits $bb'$ in the loop of Line~\ref{trick} is bounded by $1/16$. 
By a union bound argument, the probability that some $c_t$ in the sequence $c_0 \cdots c_{f(n)-1}$ is not drawn with a fair probability $\tfrac{1}{2}$ is bounded by $\tfrac{1}{16}$. Equivalently, the probability that all random randoms bits $c_t$ of the sequence $c_0 \cdots cs_{f(n)-1}$ are drawn with fair probability $\tfrac{1}{2}$ is at least $\tfrac{15}{16}$. 

To safely estimate the error probability of $\mathcal{N}$, we restrict ourselves to situations when $\mathcal{M}$ and $\mathcal{N}$ behave the same, and assume that $\mathcal{N}$ always makes errors otherwise. These situations happen when $\mathcal{M}$ and $\mathcal{N}$ use the same advice as well as the same fair probability for their random processes. 
%If we restrict to random drawings of $\mathcal{N}$ such that $\mathcal{N}$ has properly computed $\bar p$ and the sequence $c_0,\dots,\cdots c_{f(n)-1}$, then it behaves exactly as $\mathcal{M}$. 
These two events are independent and of probability at least $\tfrac{9}{10}$ and at least $\tfrac{15}{16}$, respectively. Hence, $\mathcal{M}$ and $\mathcal{N}$ agree on any input $w$ with probability at least $\tfrac{9}{10} \cdot \tfrac{15}{16} > \tfrac{4}{5}$. Consequently, the probability that $\mathcal{N}$ decides correctly whether $w \in L$ or not is bounded by $\tfrac{4}{5} \cdot \tfrac{2}{3} > \tfrac{1}{2}$. As before, this probability can be made larger than $\tfrac{2}{3}$ by repeating Procedure~\ref{algo4} a constant number of times and taking the majority of the decisions as output~\cite{AroraBarak06}. This shows that $L(\mathcal{N}) = L(\mathcal{M}) = L$, and therefore $L \in \mathbf{SNN} \left[ p,O(f^2) \right]$.  
\end{proof}

\begin{algorithm}[h!]
\DontPrintSemicolon
\SetKwInOut{Input}{Inputs}
\smallskip
\KwInput{input $w \in \Sigma^n$}
\For{$i = 0, \dots, k(n) := \ulcorner 10 p (1-p) f^2(n) \urcorner$}{
    {Draw a random bit $b_i$ with probability $p$}\;
    }
{Compute the estimation of the advice of $\mathcal{M}$ $\bar p' = p'_0 \cdots p'_{\log(f(n))-1} = \delta_2^{-1} (\frac{1}{k} \sum_{i=0}^{k(n)-1} b_i)[0:\log(f(n))-1]$}\;
\For{$t = 0, \dots, f(n)-1$}{
    {$c_t = 0$;}\;
    \For{$i = 0, \dots, \ulcorner \frac{- 4 - \log(f(n))}{\log(p^2+(1-p)^2)}\urcorner$}{ \label{trick}
        {Draw $2$ random bits $b$ and $b'$ with probability $p$}\;
        \lIf{$bb' = 01$}{$c_t = 0$; break}
        \lIf{$bb' = 10$}{$c_t = 1$; break}
        %\lElse{pass}
    }
}
\For{$t = 0, \dots, f(n)-1$}{
    {Simulate the PTM/log(A) $\mathcal{M}$ using the advice
    % \begin{equation*}
        $\bar p' = p'_0 \cdots p'_{\log(f(n))-1}$
    % \end{equation*}
    and the sequence of random choices
    % \begin{equation*}
        $\bar c = c_0 \cdots c_{f(n)-1}$
    % \end{equation*}
    }\;
}
\Return Output of $\mathcal{M}$ over $w$ at time step $f(n)$
\caption{}
\label{algo4}
\end{algorithm}

% PARAGRAPH %
The class of languages decided in polynomial time by stochastic networks using real probabilities and Turing machines using related advices are the same. In this case, however, the length of the advice are logarithmic instead of polynomial.

\begin{corollary}
\label{SNN_cor}
Let $p \in \Delta$ be some real probability and $P \subseteq \Delta$ be some set of real probabilities. 
\begin{enumerate}[label=(\roman*),itemsep=0pt]
\item \label{SNN_cor_c1}
$\mathbf{SNN}[p, \mathrm{poly}] = \mathbf{PTMA}[\alpha(\bar p, \mathrm{log}), \mathrm{poly}]$.
\item \label{SNN_cor_c2}
$\mathbf{SNN}[P, \mathrm{poly}] = \mathbf{PTMA}[\alpha(\bar P, \mathrm{log}), \mathrm{poly}]$.
\end{enumerate}
\end{corollary}

\begin{proof}
%\ref{SNN_cor_c1} Let $L \in \mathbf{SNN}[p, \mathrm{poly}]$. Then, there exists $f \in \mathrm{poly}$ such that $L \in \mathbf{SNN}[p, f]$. By Proposition~\ref{SNN_prop}--\ref{SNN_prop_c1}, $L \in \mathbf{PTMA}[\alpha(\bar p, c\mathrm{log}(f)), g \circ f]$, for some $c > 0$ and some $g \in \mathrm{poly}$. Thus $L \in \mathbf{PTMA}[\alpha(\bar p, \mathrm{log}), \mathrm{poly}]$. Conversely, let $L \in \mathbf{PTMA}[\alpha(\bar p, \mathrm{log}), \mathrm{poly}]$. Then, there exist $f' \in \mathrm{log}$ and $f \in \mathrm{poly}$ such that $L \in \mathbf{PTMA}[\alpha(\bar p, f'), f]$. By Proposition~\ref{SNN_prop}--\ref{SNN_prop_c2}, $L \in \mathbf{SNN}[p, g \circ f]$, for some $g \in \mathrm{poly}$. Therefore $L \in \mathbf{SNN}[p, \mathrm{poly}]$.
%
%\ref{SNN_cor_c2} This point follows directly from point~\ref{SNN_cor_c1} by taking the union over all $p \in P$.
The proof is similar to that of Corollary~\ref{ANN_cor}.
\end{proof}

% ======= %
% SECTION %
% ======= %
\section{Hierarchies}
\label{sec:results}

% PARAGRAPH %
In this section, we provide a refined characterization of the super-Turing computational power of analog, evolving, and stochastic neural networks based on the Kolmogorov complexity of their real weights, evolving weights, and real probabilities, respectively. More specifically, we show the existence of infinite hierarchies of classes of analog and evolving neural networks located between $\mathbf{P}$ and $\mathbf{P/poly}$. We also establish the existence of an infinite hierarchy of classes of stochastic neural networks between $\mathbf{BPP}$ and $\mathbf{BPP/log^*}$. Beyond proving the existence and providing examples of such hierarchies, we describe a generic way of constructing them based on classes of functions of increasing complexity.

% PARAGRAPH %
Towards this purpose, we define the {\it Kolmogorov complexity} of a real number as stated in a related work~\cite{SiegelmannEtAl97}. Let $\mathcal{M}_U$ be a universal Turing machine, $f, g : \mathbb{N} \rightarrow \mathbb{N}$ be two functions, and $\alpha \in \Sigma^\omega$ be some infinite word. We say that $\alpha \in \bar K^f_g$ if there exists $\beta \in \Sigma^{\omega}$ such that, for all but finitely many $n$, the machine $\mathcal{M}_U$ with inputs $\beta[0:m-1]$ and $n$ will output $\alpha[0:n-1]$ in time $g(n)$, for all $m \geq f(n)$. In other words, $\alpha \in \bar K^f_g$ if its $n$ first bits can be recovered from the $f(n)$ first bits of some $\beta$ in time $g(n)$. The notion expressed its interest when $f(n) \leq n$, in which case $\alpha \in \bar K^f_g$ means that every $n$-long prefix of $\alpha$ can be compressed into and recovered from a smaller $f(n)$-long prefix of $\beta$. Given two classes of functions $\mathcal{F}$ and $\mathcal{G}$, we define $\bar K^{\mathcal{F}}_{\mathcal{G}} = \bigcup_{f \in \mathcal{F}} \bigcup_{g \in \mathcal{G}} \bar K^f_g$. Finally, for any real number $r \in \Delta$ with associated binary expansion $\bar r = \delta^{-1}_4(r) \in \Sigma^{\omega}$, we say that $r \in K^f_g$ (resp.~$r \in K^{\mathcal{F}}_{\mathcal{G}}$) iff $\bar r \in \bar K^f_g$ (resp.~$\bar r \in \bar K^{\mathcal{F}}_{\mathcal{G}}$). 

% PARAGRAPH %
Given some set of functions $\mathcal{F} \subseteq \mathbb{N}^\mathbb{N}$, we say $\mathcal{F}$ is a class of \emph{reasonable advice bounds} if the following conditions hold:
\begin{itemize}[itemsep=0pt]
    \item Sub-linearity: for all $f \in \mathcal{F}$, then $f(n) \leq n$ for all $n \in \mathbb{N}$.
    \item Dominance by a polynomially computable function: for all $f \in \mathcal{F}$, there exists $g \in  \mathcal{F}$ such that $f \leq g$ and $g$ is computable in polynomial time.
    \item Closure by polynomial composition on the right: For all $f \in \mathcal{F}$ and for all $p \in \mathrm{poly}$, there exist $g \in \mathcal{F}$ such that $f \circ p \leq g$.
\end{itemize}

% PARAGRAPH %
For instance, $\mathrm{log}$ is a class of reasonable advice bounds. All properties in this definition are necessary for our separation theorems. The first and second conditions are necessary to define Kolmogorov reals associated to advices of bounded size. The third condition comes from the fact that RNNs can access any polynomial number of bits from their weights during polynomial time of computation. Note that our definition is slightly weaker than that of Balc\'azar et al., who further assume that the class should be closed under $O(.)$~\cite{SiegelmannEtAl97}.

% PARAGRAPH %
The following theorem relates non-uniform complexity classes, based on polynomial time of computation $\mathbf{P}$ and reasonable advice bounds $\mathcal{F}$, with classes of analog and evolving networks using weights inside $K_{\mathrm{poly}}^\mathcal{F}$ and $\bar K_{\mathrm{poly}}^\mathcal{F}$, respectively.

\begin{theorem}
\label{hierarchy_thm1}
Let $\mathcal{F}$ be a class of reasonable advice bounds, and let $K_{\mathrm{poly}}^\mathcal{F} \subseteq \Delta$ and $\bar K_{\mathrm{poly}}^\mathcal{F} \subseteq \Sigma^\omega$ be the sets of Kolmogorov reals associated with $\mathcal{F}$. Then
$$
\mathbf{P} / \mathcal{F}^* = \mathbf{ANN} \left[ K_{\mathrm{poly}}^\mathcal{F}, \mathrm{poly} \right] = \mathbf{ENN} \left[ \bar K_{\mathrm{poly}}^\mathcal{F}, \mathrm{poly} \right] .
$$
\end{theorem}

\begin{proof}
We prove the first equality. By definition, $\mathbf{P} / \mathcal{F}^*$ is the class of languages decided in polynomial time by some TM/A using any possible prefix advice of length $f \in \mathcal{F}$, namely,
$$
\mathbf{P} / \mathcal{F}^* = \mathbf{TMA} \left[ \alpha \big( \Sigma^\omega, \mathcal{F} \big), \mathrm{poly} \right] .
$$
In addition, Corollary~\ref{ANN_cor} ensures that
$$
\mathbf{ANN} \left[ K_{\mathrm{poly}}^\mathcal{F}, \mathrm{poly} \right] = \mathbf{TMA} \left[ \alpha \big( \bar K_{\mathrm{poly}}^\mathcal{F}, \mathrm{poly} \big), \mathrm{poly} \right] .
$$ 
Hence, we need to show that
\begin{equation}
\label{desired_property_1}
\mathbf{TMA} \left[ \alpha \big( \Sigma^\omega, \mathcal{F} \big), \mathrm{poly} \right] = \mathbf{TMA} \left[ \alpha \big( \bar K_{\mathrm{poly}}^\mathcal{F}, \mathrm{poly} \big), \mathrm{poly} \right] .
\end{equation}
 % We show the backward inclusion of Eq.~(\ref{desired_property_1}). Let  $L \in \mathbf{TMA} \left[ \alpha(\bar K_{\mathrm{poly}}^\mathcal{F}, \mathrm{poly}), \mathrm{poly} \right]$ a language decided by 
 % some TM/A $\mathcal{M}$ with advice $\alpha(\bar r, p)$ with its advice function $a \in \alpha(\bar K_{\mathrm{poly}}^\mathcal{F}, \mathrm{poly})$. By definition, $a$ can be computed from another advice function $b \in \alpha(\Sigma^\omega, \mathcal{F})$ in polynomial time. To simulate $M$ with advice $a$, we consider the machine $M'$ which on advice $b$ first construct $a$ in polynomial time then does the same computation as $M$
 % in the same time bound. By construction, $M'$ belongs to $\mathbf{TMA} \left[ \alpha(\Sigma^\omega, \mathcal{F})\right]$, which proves the inclusion.
\noindent Equation~\ref{desired_property_1} can be understood as follows: in polynomial time of computation, the TM/As using small advices (of size $\mathcal{F}$) are equivalent to those using larger but compressible advices (of size $\mathrm{poly}$ and inside $\bar K_{\mathrm{poly}}^\mathcal{F}$).
 
For the sake of simplicity, we suppose that the polynomial time of computation of the TM/As are clear from the context by introducing the following abbreviations:
\begin{eqnarray*}
\mathbf{TMA} \left[ \alpha \big( \Sigma^\omega, \mathcal{F} \big), \mathrm{poly} \right] & := & \mathbf{TMA} \left[ \alpha \big( \Sigma^\omega, \mathcal{F} \big) \right] \\
\mathbf{TMA} \left[ \alpha \big( \bar K_{\mathrm{poly}}^\mathcal{F}, \mathrm{poly} \big), \mathrm{poly} \right] & := & \mathbf{TMA} \left[ \alpha \big( \bar K_{\mathrm{poly}}^\mathcal{F}, \mathrm{poly} \big) \right] .    
\end{eqnarray*}

We show the backward inclusion of Eq.~(\ref{desired_property_1}). Let 
$$
L \in \mathbf{TMA} \left[ \alpha \big( \bar K_{\mathrm{poly}}^\mathcal{F}, \mathrm{poly} \big) \right] .
$$
Then, there exists some TM/A $\mathcal{M}$ using advice $\alpha(\bar r, p_1)$, where $\bar r \in \bar K_{\mathrm{poly}}^\mathcal{F}$ and $p_1 \in \mathrm{poly}$, deciding $L$ in time $p_2 \in \mathrm{poly}$. Since $\bar r \in \bar K_{\mathrm{poly}}^\mathcal{F}$, there exist $\beta \in \Sigma^\omega$, $f \in \mathcal{F}$ and $p_3 \in \mathrm{poly}$ such that the $p_1(n)$ bits of $\bar r$ can be computed from the $f(p_1(n))$ bits of $\beta$ in time $p_3(p_1(n))$. Hence, the TM/A $\mathcal{M}$ can be simulated by the TM/A $\mathcal{M}'$ with advice $\alpha(\beta, f \circ p_1)$ working as follows: on every input $w \in \Sigma^n$, $\mathcal{M}'$ first queries its advice string $\beta_0 \beta_1 \cdots \beta_{f(p_1(n))-1}$, then reconstructs the advice $r_0 r_1 \dots r_{p_1(n)-1}$ in time $p_3(p_1(n))$, and finally simulates the behavior of $\mathcal{M}$ over input $w$ in real time. Clearly, $L(\mathcal{M}') = L(\mathcal{M}) = L$. In addition, $p_3 \circ p_1 \in \mathrm{poly}$, and since $\mathcal{F}$ is a class of reasonable advice bounds, there is $g \in \mathcal{F}$ such that $f \circ p_1 \leq g$. Therefore,
%$$
%L \in \mathbf{TMA} \left[ \alpha ( \beta, g), O(p_3 \circ p_1 + p_2) \right] \subseteq \mathbf{TMA} \left[ \alpha \big( \Sigma^\omega, \mathcal{F} \big) \right] .
%$$
$$
L \in \mathbf{TMA} \left[ \alpha ( \beta, g) \right] \subseteq \mathbf{TMA} \left[ \alpha \big( \Sigma^\omega, \mathcal{F} \big) \right] .
$$

We now prove the forward inclusion of Eq.~(\ref{desired_property_1}). Let 
$$
L \in \mathbf{TMA} \left[ \alpha \big( \Sigma^\omega, \mathcal{F} \big) \right] .
$$
Then, there exists some TM/A $\mathcal{M}$ with advice $\alpha(\bar r, f)$, with $\bar r \in\Sigma^\omega$ and $f \in \mathcal{F}$, deciding $L$ in time $p_1 \in \mathrm{poly}$. Since $\mathcal{F}$ is a class of reasonable advice bounds, there exists $g \in \mathcal{F}$ such that $f \leq g$ and $g$ is computable in polynomial time. We now define $\bar s \in \bar K_{\mathrm{poly}}^\mathcal{F}$ using $\bar r$ and $g$ as follows:
%Let $\bar r_i$ be the subword of $\bar r$ between positions $g(i-1)$ and $g(i)$ when $i>0$ and between position $0$ and $g(0)$ for $i=0$. 
for each $i \geq 0$, let $\bar r_i$ be the sub-word of $\bar r$ defined by
$$
\bar r_i = 
\begin{cases}
    r_0 r_1 \cdots r_{g(0)-1} & \text{if $i = 0$} \\
    r_{g(i-1)} r_{g(i-1)+1} \cdots r_{g(i)-1} & \text{if $i > 0$} \\
\end{cases}
$$
%We define the word $\bar s \in \Sigma^{\omega}$ as
and let
$$ 
\bar s = \bar r_0 0 \bar r_1 0 \bar r_2 0 \cdots .
$$ 
Given the $g(n)$ first bits of $\bar r$, we can build the $g(n) + n \geq n$ first bits of $\bar s$ by computing  $g(i)$ and the corresponding block $\bar r_i$ (which can be empty) for all $i \leq n$, and then intertwining those with $0$'s. This process can be done in polynomial time, since $g$ is computable is polynomial time. Therefore, $\bar s \in \bar K_{\mathrm{poly}}^\mathcal{F}$. 

% Since $\mathcal{F}$ is a class of reasonable advice bounds, there exists $g \in \mathcal{F}$ such that $g$ is computable in time $p_2 \in \mathrm{poly}$ and $f \leq g \leq p$ for some $p \in \mathrm{poly}$. Let $q \in \mathrm{poly}$ be defined by $q(n) = (n+1)p(n)$. We have
% \begin{align*}
% \delta_n := g(n+1) - g(n) & \leq g(n+1) \leq p(n+1) \\
% & = (n+2)p(n+1) - (n+1)p(n+1) \\
% & \leq (n+2)p(n+1) - (n+1)p(n) = q(n+1) - q(n) .
% \end{align*}
% Now, assume that $g(0) = q(0) = 0$, and consider the infinite word $\bar s \in \Sigma^\omega$ consisting of chunks of bits of $\bar r$ interspersed with blocks of $0$'s, as illustrated in Figure~\ref{fig:advice}, and defined by
% \begin{align*}
% & s_{q(i)} s_{q(i) + 1} \cdots s_{q(i) + \delta_i} = r_{g(i)} r_{g(i) + 1} \cdots s_{g(i+1)} \\
% & s_{q(i) + \delta_i + 1} s_{q(i) + \delta_i + 1} \cdots s_{q(i+1) - 1} = 00 \cdots 0
% \end{align*}
% for all $i \geq 0$. Given $g(n)$ bits of $\bar r$, we can compute $q(n) \geq n$ bits of $\bar s$ in polynomial time as follows: first, compute the successive values $g(0), g(1), \dots, g(n)$ and $q(0), q(1), \dots, q(n)$, each of them and thus all of them in polynomial time, and then, for each $0 \leq i < n$, copy the following block of $g(i)$ bits of $\bar r$ and concatenate it with a block of $0$'s until reaching position $q(i+1)$. Since $g \in \mathcal{F}$, it follows that $\bar s \in \bar K_{\mathrm{poly}}^\mathcal{F}$. 
Let $q(n) = 2n$. Since $\mathcal{F}$ is a class of reasonable advice bounds, $g(n)\leq n$, and thus $q(n) = 2n \geq g(n) + n$. Now, consider the TM/A $\mathcal{M}'$ with advice $\alpha(\bar s, q)$ working as follows. On every input $w \in \Sigma^n$, the machine $\mathcal{M}'$ first queries its advice $\alpha(\bar s, q)(n) = s_0 s_1 \cdots s_{q(n) - 1}$. Then, $\mathcal{M}'$ reconstructs the string $r_0 r_1 \cdots r_{g(n)-1}$ by computing $g(i)$ and then removing $n$ $0$'s from $\alpha(\bar s, q)(n)$ at positions $g(i) + i$, for all $i\leq n$. This is done in polynomial time, since $g$ is computable in polynomial time. Finally, $\mathcal{M}'$ simulates $\mathcal{M}$ with advice $r_0 r_1 \cdots r_{g(n)-1}$ in real time. Clearly, $L(\mathcal{M}') = L(\mathcal{M}) = L$. Therefore,
$$
L \in \mathbf{TMA} \left[ \alpha(\bar s, q) \right] \subseteq \mathbf{TMA} \left[ \alpha \big( \bar K_{\mathrm{poly}}^\mathcal{F}, \mathrm{poly} \big) \right] .
$$

The property that have just been established together with Corollary~\ref{ENN_cor} proves the second equality.
\end{proof}

% \begin{figure}[t]
%     \centering
%     \includegraphics[width=10cm]{figures/advice.pdf}
%     \caption{Construction of the infinite word $\bar s$ from $\bar r$.}
%     \label{fig:advice}
% \end{figure}

% PARAGRAPH %
We now prove the analogous of Theorem~\ref{hierarchy_thm1} for the case of probabilistic complexity classes and machines. In this case, however, the class of advice bounds does not anymore correspond exactly to the Kolmogorov space bounds of the real probabilities. Instead, a logarithmic correcting factor needs to be introduced. Given some class of functions $\mathcal{F}$, we let $\mathcal{F} \circ \log$ denote the set $\lbrace f \circ \log \mid f \in \mathcal{F}\rbrace$.

% Note that in this case, the class of length functions involved in the Kolmogorov class is $2^{\mathcal{F}}$ instead of $\mathcal{F}$.  

\begin{theorem}
\label{hierarchy_thm2}
Let $\mathcal{F}$ be a class of reasonable advice bounds, %$K_{\mathrm{poly}}^\mathcal{F} \subseteq \Delta$ be the set of Kolmogorov reals associated with $\mathcal{F}$. Then
then 
$$
\mathbf{BPP} / (\mathcal{F} \circ \log)^* = \mathbf{SNN} \left[ K_{\mathrm{poly}}^{\mathcal{F}}, \mathrm{poly} \right] .
$$
\end{theorem}

\begin{proof}
%The proof resembles that of Theorem~\ref{hierarchy_thm1}. 
By definition, $\mathbf{BPP} / (\mathcal{F} \circ \log)^*$ is the class of languages decided by PTM/A using any possible prefix advice of length $f \in \mathcal{F} \circ \log$:
$$
\mathbf{BPP} / (\mathcal{F} \circ \log)^* = \mathbf{PTMA} \left[ \alpha \big( \Sigma^\omega, \mathcal{F} \circ \log \big), \mathrm{poly} \right] .
$$
% By the exact same proof as in Theorem~\ref{hierarchy_thm1}, except that the machines we simulate at the end are stochastic, 
% we have: 

% $$\mathbf{PTMA} \left[ \alpha  \big( \bar K_{\mathrm{poly}}^{\mathcal{F} \circ \log}, \mathrm{poly} \big), \mathrm{poly} \right] =
% \mathbf{PTMA} \left[ \alpha \big( \Sigma^\omega, \mathcal{F} \circ \log \big), \mathrm{poly} \right] $$

\noindent Moreover, Corollary~\ref{SNN_cor} ensures that
$$
\mathbf{SNN} \left[ K_{\mathrm{poly}}^{\mathcal{F}}, \mathrm{poly} \right] = \mathbf{PTMA} \left[ \alpha  \big( \bar K_{\mathrm{poly}}^{\mathcal{F}}, \mathrm{log} \big), \mathrm{poly} \right] .
$$ 
Hence, we need to prove the following equality:
\begin{equation} \label{eq3}
\mathbf{PTMA} \left[ \alpha  \big( \bar K_{\mathrm{poly}}^{\mathcal{F}}, \mathrm{log} \big), \mathrm{poly} \right] = \mathbf{PTMA} \left[ \alpha \big( \Sigma^\omega, \mathcal{F} \circ \log \big), \mathrm{poly} \right] 
\end{equation}

We first prove the forward inclusion of Eq.~\ref{eq3}. Let 
$$
L \in \mathbf{PTMA} \left[ \alpha  \big( \bar K_{\mathrm{poly}}^{\mathcal{F}}, \mathrm{log} \big), \mathrm{poly} \right] .
$$
Then, there exists some PTM/A $\mathcal{M}$ using advice $\alpha(\bar r, c\log)$, where $\bar r \in \bar K_{\mathrm{poly}}^\mathcal{F}$ and $c > 0$, that decides $L$ in time $p_1 \in \mathrm{poly}$. Since $\bar r \in \bar K_{\mathrm{poly}}^\mathcal{F}$, there exist $\beta \in \Sigma^\omega$ and $f \in \mathcal{F}$ such that $\bar r[0: n-1]$ can be computed from $\beta[0: f(n)-1]$ in time $p_2(n) \in \mathrm{poly}$, for all $n \geq 0$. Consider the PTM/A $\mathcal{M}'$ with advice $\alpha(\beta, f \circ c\log)$ working as follows. First, $\mathcal{M}'$ queries its advice $\beta[0: f (c\log(n))-1]$, then it computes $\bar r[0: c\log(n)-1]$ from this advice in time $p_2(\log(n))$, and finally it simulates $\mathcal{M}$ with advice $\bar r[0: c\log(n)-1]$ in real time. Consequently, $\mathcal{M}'$ decides the same language $L$ as $\mathcal{M}$, and works in time $O(p_1 + p_2 \circ \log) \in \mathrm{poly}$. Therefore, 
$$
L \in \mathbf{PTMA} \left[ \alpha \big( \Sigma^\omega, \mathcal{F} \circ \log \big), \mathrm{poly} \right] .
$$
%which decodes the first $f(c\log(n))$ bits of its advice $\beta$ as $\mathcal{N}$ to obtain the first $c\log(n)$ bits of $\bar r$ and then computes as $\mathcal{M}$ using these bits of $\bar r$ as advice. The machine $\mathcal{M}'$ works in time $O(p_1+p_2) \in \mathrm{poly}$ and uses and advice of size $f(c\log(n)) \in  \mathcal{F} \circ \log$. Since it simulates $\mathcal{M}$, it decides $L$, which proves the inclusion.

We now prove the backward inclusion of Eq.~\ref{eq3}. Let
$$
L \in \mathbf{PTMA} \left[ \alpha \big( \Sigma^\omega, \mathcal{F} \circ \log \big), \mathrm{poly} \right] .
$$
Then, there exists some PTM/A $\mathcal{M}$ using advice $\alpha(\bar r, f \circ c\log)$, where $\bar r \in \Sigma^\omega$, $f \in \mathcal{F}$ and $c > 0$, that decides $L$ in time $p_1 \in \mathrm{poly}$. Using the same argument as in the proof of Theorem~\ref{hierarchy_thm1}, there exist $\bar s \in \bar K_{\mathrm{poly}}^{\mathcal{F}}$ and $g \in \mathcal{F}$ such that $f \leq g$ and the smaller word $\bar r[0 : g(n) -1]$ can be retrieved from the larger one $\bar s[0 : 2n -1]$ in time $p_2(n) \in \mathrm{poly}$, for all $n \geq 0$. 
%\begin{itemize}
%    \item $\bar s \in \bar K_{\mathrm{poly}}^{\mathcal{F}}$
%    \item from the first $2n$ bits of $\bar s$, we can build the first $g(n)$ bits of $\bar r$ in time
%    $p_2 \in \mathrm{poly}$
%\end{itemize}
Now, consider the PTM/A $\mathcal{M}'$ using advice $\alpha(\bar s, 2c\log)$ and working as follows. First, $\mathcal{M}'$ queries its advice $\bar s [0: 2c\log(n)-1]$, then it reconstructs $\bar r [0: g(c\log(n))-1]$ from this advice in time $O(p_2(\log(n)))$, and finally, it simulates $\mathcal{M}$ with advice $\bar r [0: g(c\log(n))-1]$ in real time. Since $\bar r [0: g(c\log(n))-1]$ extends $\bar r [0: f(c\log(n))-1]$, $\mathcal{M}'$ and $\mathcal{M}$ decide the same language $L$. In addition, $\mathcal{M}'$ works in time $O(p_1 + p_2\circ \log ) \in \mathrm{poly}$. Therefore, 
$$
\mathbf{PTMA} \left[ \alpha  \big( \bar K_{\mathrm{poly}}^{\mathcal{F}}, \mathrm{log} \big), \mathrm{poly} \right].
$$
\end{proof}

% PARAGRAPH % 
We now prove the separation of non-uniform complexity classes of the form $\mathcal{C}/f$. Towards this purpose, we assume that each class of languages $\mathcal{C}$ is defined on the basis of a set of machines that decide these languages. For the sake of simplicity, we naturally identify $\mathcal{C}$ with its associated class of machines. 
%the class of languages $\mathcal{C}$ will be naturally identified with the set of machines deciding them. 
For instance, $\mathbf{P}$ and $\mathbf{BPP}$ are identified with the set of Turing machines and probabilistic Turing machines working in polynomial time, respectively. In this context, we show that, as soon as the advice is increased by a single bit, the capabilities of the corresponding machines are also increased. To achieve this result, the two following weak conditions are required. First, $\mathcal{C}$ must contain the machines capable of reading their full inputs (of length $n$) and advices (of length $f(n)$) (otherwise, any additional advice bit would not change anything). Hence, $\mathcal{C}$ must at least include the machines working in time $O(n + f(n))$. Secondly, the advice length $f(n)$ should be smaller than $2^n$, for otherwise, the advice could encode any possible language, and the corresponding machine would have full power. The following result is proven for machines with general (i.e., non-prefix) advices, before being stated for the particular case of machines with prefix advices. 
%We prove that as soon that we give a single bit more of advice, we can decide new languages. To make this separation theorem true, we need two very mild conditions:
%\begin{itemize}
%    \item The machines must be able to read the input and the whole advice (otherwise adding a single bit would not change anything). That is $\mathcal{C}$ must contain the Turing machines with advice working in time $O(f(n) + n)$. 
%    \item The additional bit should be useful, that is $f(n)$ bits should not be enough to encode all words of size $n$ of the language. Hence, $f(n) < 2^n$.
%\end{itemize}  

\begin{theorem}
\label{th:diag}
Let $f, g : \mathbb{N} \rightarrow \mathbb{N}$ be two increasing functions such that $f(n) < g(n) \leq 2^n$, for all $n \in \mathbb{N}$. 
Let $\mathcal{C}$ be a set of machines containing the Turing machines working in time $O(n + f(n))$. Then $\mathcal{C}/f \subsetneq \mathcal{C}/g$.
\end{theorem}

\begin{proof}
Any Turing machine $M$ with advice $\alpha$ of size $f$ can be simulated by some Turing machine $M'$ with an advice $\alpha'$ of size $g > f$. Indeed, take $\alpha'(n) = (1)^{g(n)-f(n)-1} 0 \alpha(n)$. Then, on any input $w \in \Sigma^n$, the machine $M'$ queries its advice $\alpha'(n)$, erases all $1$'s up to and including the first encountered $0$, and then simulates $M$ with advice $\alpha(n)$. Clearly, $M$ and $M'$ decide the same language.

To prove the strictness of the inclusion, we proceed by diagonalization. Recall that the set of (probabilistic) Turing machines is countable. Let $M_0, M_1, M_2, \dots$ be an enumeration of the machines in $\mathcal{C}$. For any $M_k \in \mathcal{C}$ and any advice $\alpha : \mathbb{N} \rightarrow \Sigma^*$ of size $f$, let $M_k / \alpha$ be the associated (probabilistic) machine with advice, and let $L( M_k / \alpha)$ be its associated language.
%Every 2-tape Turing machine $M_k \in \mathcal{C}$ and every advice function $\alpha : \mathbb{N} \rightarrow \Sigma^*$ is naturally associated with the Turing machine with advice $M_k / \alpha$ which, for each input of length $n \in \mathbb{N}$, first writes the advice string $\alpha(n)$ on its first tape and then continues the computation according to its program. Let $L(M_k / \alpha)$ be the language decided by $M_k / \alpha$. 
The language $L(M_k / \alpha)$ can be written as the union of its sub-languages of words of length $n$, i.e.
$$
L(M_k / \alpha) = \bigcup_{n \in \mathbb{N}} L(M_k / \alpha)^n.
$$
For each $k, n \in \mathbb{N}$, consider the set of sub-languages of words of length $n$ decided by $M_k / \alpha$, for all possible advices $\alpha$ of size $f$, i.e.:
$$
\mathcal{L}^n_k = \big\{ L(M_k / \alpha)^n : \alpha \text{ is an advice of size $f$} \big\}.
$$
Since there are at most $2^{f(n)}$ advice strings of length $f(n)$, it follows that $| \mathcal{L}^n_k | \leq 2^{f(n)}$, for all $k \in \mathbb{N}$, and in particular, that $| \mathcal{L}^n_n | \leq 2^{f(n)}$. By working on the diagonal $\mathcal{D} = \big( \mathcal{L}^n_n \big)_{n \in \mathbb{N}}$ of the sequence $\big( \mathcal{L}^n_k \big)_{k, n \in \mathbb{N}}$ (illustrated in Table~\ref{tab_diagonal}), we will build a language $A = \bigcup_{n \in \mathbb{N}} A^n$ that cannot be decided by any Turing machine in $\mathcal{C}$ with advice of size $f$, but can be decided by some Turing machine in $\mathcal{C}$ with advice of size $g$. It follows that $A \in (\mathcal{C}/g) \setminus (\mathcal{C}/f)$, and therefore, $\mathcal{C}/f \subsetneq \mathcal{C}/g$.

\renewcommand{\arraystretch}{1.5}
\begin{table}[h!]
    \centering
    \begin{tabular}{c | cccccc}
       & $0$ & $1$ & $2$ & $\cdots$ & $n$ & $\cdots$ \\ 
    \hline
    $0$ & \hl{$\mathcal{L}^0_0$} & $\mathcal{L}^1_0$ & $\mathcal{L}^2_0$ & $\cdots$ & $\mathcal{L}^n_0$ & $\cdots$ \\ 
    $1$ & $\mathcal{L}^0_1$ & \hl{$\mathcal{L}^1_1$} & $\mathcal{L}^2_1$ & $\cdots$ & $\mathcal{L}^n_1$ & $\cdots$ \\ 
    $2$ & $\mathcal{L}^0_2$ & $\mathcal{L}^1_2$ & \hl{$\mathcal{L}^2_2$} & $\cdots$ & $\mathcal{L}^n_2$ & $\cdots$ \\ 
    $\vdots$ & $\vdots$ & $\vdots$ & $\vdots$ & \hl{$\ddots$} & $\vdots$ & $\cdots$ \\ 
    $n$ & $\mathcal{L}^0_n$ & $\mathcal{L}^1_n$ & $\mathcal{L}^2_n$ & $\cdots$ & \hl{$\mathcal{L}^n_n$} & $\cdots$ \\ 
    $\vdots$ & $\vdots$ & $\vdots$ & $\vdots$ & $\cdots$ & $\vdots$ & \hl{$\ddots$} \\
    \end{tabular}
    \caption{Illustration of the sequence $\big( \mathcal{L}^n_k \big)_{k, n \in \mathbb{N}}$. Each set $\mathcal{L}^n_k$ satisfies $| \mathcal{L}^n_k | \leq 2^{f(n)}$. The diagonal $\mathcal{D} = \big( \mathcal{L}^n_n \big)_{n \in \mathbb{N}}$ is highlighted.}
    \label{tab_diagonal}
\end{table}

Let $n \in \mathbb{N}$. For each $i < 2^n$, let $b(i) \in \Sigma^n$ be the binary representation of $i$ over $n$ bits. For any subset $\mathcal{L} \subseteq \mathcal{L}^n_n$, let
$$
\mathcal{L} \big( b(i) \big) = \left\{ L \in \mathcal{L} : b(i) \in L \right\} \text{ ~and~ } \bar{\mathcal{L}} \big(b(i) \big) = \left\{ L \in \mathcal{L} : b(i) \not \in L \right\}.
$$
%If $\mathcal{L}^n_{n} = \{ \emptyset \}$, define $A^n_{f(n)+1} = \{ b(0) \}$. Obviously, $A^n_{f(n)+1} \not \in \mathcal{L}^n_{n}$. Otherwise,
Consider the sequence $( \mathcal{L}^n_{n,0}, \dots, \mathcal{L}^n_{n, f(n)+1} )$ of decreasing subsets of $\mathcal{L}^n_{n}$ and the sequence $( A^n_0, \dots, A^n_{f(n)+1} )$ of sub-languages of words of length $n$ defined by induction for every $0 \leq i \leq f(n)$ as follows
\begin{align*}
\mathcal{L}^n_{n,0} = \mathcal{L}^n_n \text{ ~and~ }
\mathcal{L}^n_{n,i+1} & = 
\begin{cases}
\mathcal{L}^n_{n,i} \big( b(i) \big) & \text{if } | \mathcal{L}^n_{n,i} \big( b(i) \big) | < | \bar{\mathcal{L}}^n_{n,i} \big( b(i) \big) | \\
\bar{\mathcal{L}}^n_{n,i} \big( b(i) \big) & \text{otherwise}
\end{cases}
\end{align*}
\begin{align*}
A^n_{0} = \emptyset \text{ ~and~ }
\mathcal{A}^n_{i+1} & =
\begin{cases}
\mathcal{A}^n_{i} \cup \{ b(i) \} & \text{if } | \mathcal{L}^n_{n,i} \big( b(i) \big) | < | \bar{\mathcal{L}}^n_{n,i} \big( b(i) \big) | \\
\mathcal{A}^n_{i} & \text{otherwise.}
\end{cases}    
\end{align*}
This construction is illustrated in Figure~\ref{fig_dichotomy}. 

\begin{figure}[h!]
    \centering
    \includegraphics[width=10cm]{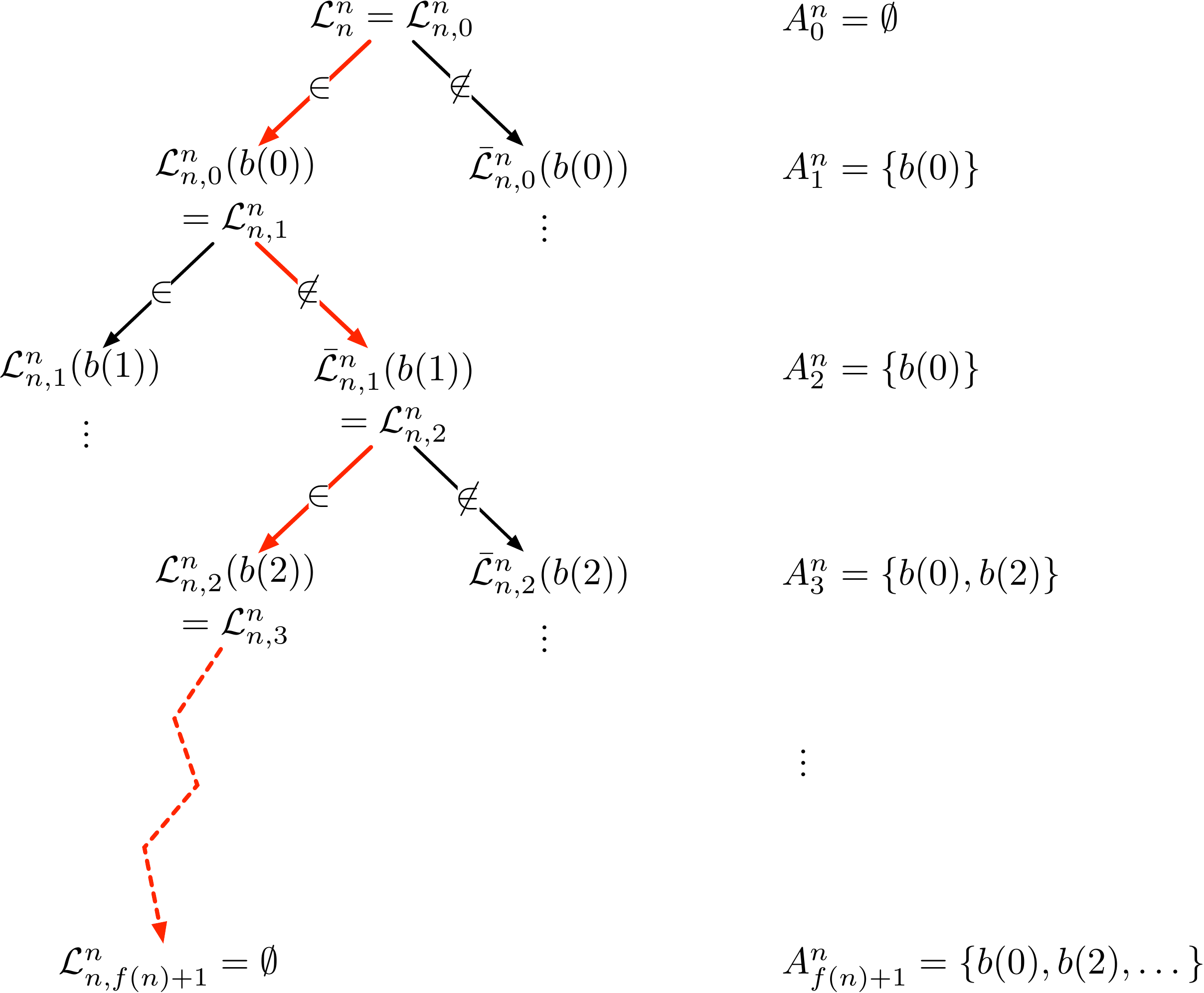}
    \caption{Inductive construction of the sequences $\big( \mathcal{L}^n_{n,i} \big)_{i=0}^{f(n)+1}$ and $\big( A^n_i \big)_{i=0}^{f(n)+1}$.}
    \label{fig_dichotomy}
\end{figure}

Note that the $n$-bit representation $b(i)$ of $i$ is well-defined, since $0 \leq i \leq f(n) < 2^n$. In addition, the construction ensures that $| \mathcal{L}^n_{n,i+1} | \leq \frac{1}{2} | \mathcal{L}^n_{n,i} |$, and since $| \mathcal{L}^n_{n,0} | = | \mathcal{L}^n_{n} | \leq 2^{f(n)}$, it follows that $| \mathcal{L}^n_{n,f(n) + 1} | = 0$, meaning that $\mathcal{L}^n_{n,f(n) + 1} = \emptyset$. Furthermore, the construction also ensure that $\mathcal{L}^n_{i} \subseteq \mathcal{L}^n_{i + 1}$ and $A^n_{f(n) + 1} \in \mathcal{L}^n_{n, i}$, for all $0 \leq i \leq f(n)$. Now, towards, a contradiction, suppose that $A^n_{f(n) + 1} \in \mathcal{L}^n_n$. %By construction, we have that $A^n_{f(n) + 1} \subseteq \{ b(i) : 0 \leq i \leq f(n) \}$ and $\mathcal{L}^n_{i} \subseteq \mathcal{L}^n_{i + 1}$ for all $0 \leq i < f(n)$, which imply that
The previous properties imply that
\begin{equation*}
A^n_{f(n) + 1} \in \bigcap_{0 \leq i \leq f(n)} \mathcal{L}^n_{n, i} = \mathcal{L}^n_{n, f(n)+1} = \emptyset    
\end{equation*}
which is a contradiction. Therefore, $A^n_{f(n) + 1} \not \in \mathcal{L}^n_n$, for all $n \in \mathbb{N}$.

Now, consider the language
$$A = \bigcup_{n \in \mathbb{N}} A^n_{f(n) + 1}.$$
By construction, $A^n_{f(n) + 1}$ is the set of words of length $n$ of $A$, meaning that $A^n_{f(n) + 1} = A^n$, for all $n \in \mathbb{N}$. We now show that $A$ cannot be decided by any machine in $\mathcal{C}$ with advice of size $f$. Towards a contradiction, suppose that $A \in \mathcal{C} / f$. Then, there exist $M_k \in \mathcal{C}$ and $\alpha : \mathbb{N} \rightarrow \Sigma^*$ of size $f$ such that $L(M_k / \alpha) = A$. On the one hand, the definition of $\mathcal{L}^k_k$ ensures that $L(M_k / \alpha)^k \in \mathcal{L}^k_k$. On the other hand, $L(M_k / \alpha)^k = A^k = A^k_{f(k) + 1} \not \in \mathcal{L}^k_k$, which is a contradiction. Therefore, $A \not \in \mathcal{C} / f$.

We now show that $A \in \mathcal{C} / g$. Consider the advice function $\alpha : \mathbb{N} \rightarrow \Sigma^*$ of size $g = f + 1$ given by
$\alpha(n) = \alpha^n_0 \alpha^n_1 \cdots \alpha^n_{f(n)}$, where
$$
\alpha^n_i = 
\begin{cases}
1 & \text{if } b(i) \in A^n \\
0 & \text{otherwise,}
\end{cases}
$$
for all $0 \leq i \leq f(n)$. Note that the advice string $\alpha(n)$ encodes the sub-language $A^n$, for all $n \in \mathbb{N}$, since the latter is a subset of $\{ b(i) : i \leq f(n) \}$ by definition. Consider the Turing machine with advice $M / \alpha$ which, on every input $w = w_0 w_1 \cdots w_{n-1}$ of length $n$, moves its advice head up to the $i$-th bit $\alpha^n_i$ of $\alpha(n)$, where $i = b^{-1}(w)$, if this bit exists (note that $i < 2^{n}$ and $| \alpha(n) | = f(n) + 1$), and accepts $w$ if and only if $\alpha^n_i = 1$. Note that these instructions can be computed in time $O(f(n) + n)$. In particular, moving the advice head up to the $i$-th bit of $\alpha(n)$ does not require to compute $i = b^{-1}(w)$ explicitly, but can be achieved by moving simultaneously the input head from the end of the input to the beginning and the advice head from left to right, in a suitable way.
It follows that
\begin{equation*}
 w \in L(M / \alpha)^n \text{ ~iff~ } \alpha^n_{b^{-1}(w)} = 1 \text{ ~iff~ } b( b^{-1} (w)) \in A^n \text{ ~iff~ } w \in A^n.
\end{equation*}
Hence, $L(M / \alpha)^n = A^n$, for all $n \in \mathbb{N}$, and thus 
$$L(M / \alpha) = \bigcup_{n \in \mathbb{N}} L(M / \alpha)^n = \bigcup_{n \in \mathbb{N}} A^n  = A.$$
Therefore, $A \in \mathcal{C} / g$. The argument can be generalized in a straightforward way to any advice size $g$ such that $f(n) + 1 \leq g(n) \leq 2^n$. 

Finally, the two properties $A \not \in \mathcal{C} / f$ and $A \in \mathcal{C} / g$ imply that $\mathcal{C} / f \subsetneq \mathcal{C} / g$.
\end{proof}

We now prove the analogous of Theorem~\ref{th:diag} for the case of machines with prefix advice. In this case, however, a stronger condition on the advice lengths $f$ and $g$ is required: $f \in o(g)$ instead of $f \leq g$.

\begin{theorem}
\label{th:diag_2}
Let $f, g : \mathbb{N} \rightarrow \mathbb{N}$ be two increasing functions such that $f \in o(g)$ and $\lim_{n \rightarrow \infty} g(n) = + \infty$. Let $\mathcal{C}$ be a set of machines containing the Turing machines working in time $O(n + f(n))$. Then $\mathcal{C}/f^* \subsetneq \mathcal{C}/g^*$.
\end{theorem}

\begin{proof}
%First, note that the requirement that $\sum_{i=0}^n f(i) \leq g(n)$ for all $n \in \mathbb{N}$ is strictly stronger than the condition $f(n) \in o ( g(n) )$. Indeed, suppose that $g(n) = \sum_{i=0}^n f(i)$. Then $f$ and $g$ obviously satisfy the requirement. In addition,
%$$
%0 \leq \lim_{n \rightarrow \infty} \frac{f(n)}{g(n)} = \lim_{n \rightarrow \infty} \frac{f(n)}{\sum_{i=0}^n f(i)} \leq \lim_{n \rightarrow \infty} \frac{f(n)}{(n+1) f(n)} = \lim_{n \rightarrow \infty} \frac{1}{(n+1)} = 0
%$$
%which shows that $\lim_{n \rightarrow \infty} \frac{f(n)}{g(n)} = 0$, and thus $f(n) \in o ( g(n) )$. On the other hand, suppose that $f(n) = 2n$ and $g(n) = n^2$. Then $\lim_{n \rightarrow \infty} \frac{f(n)}{g(n)} = 0$, meaning that that $f(n) \in o ( g(n) )$. However,
%$$
%\sum_{i=0}^n f(i) = \sum_{i=0}^n 2i = 2 \frac{n (n+1)}{2} > n^2 = g(n),
%$$
%which shows that $f$ and $g$ do not satisfy the requirement of the hypothesis.
The proof is similar to that of Theorem~\ref{th:diag}, except that we will construct the language $A$ on the basis of a sequence of integers $( n_i )_{i \in \mathbb{N}}$. Consider the sequence $(n_i)_{i \in \mathbb{N}}$ defined for all $i \geq 0$ as follows
\begin{eqnarray*}
n_{0} & = & \min \left\{ n \in \mathbb{N} : 2(f(n) + 1) \leq g(n) \right\} \\
n_{i+1} & = & \min \Big\{ n \in \mathbb{N} : 2 \sum_{j = 0}^{i} (f(n_j) + 1) + 2(f(n) + 1) \leq g(n) \Big\}.
\end{eqnarray*}
We show by induction that the sequence $( n_i )_{i \in \mathbb{N}}$ is well-defined. Since $f \in o(g)$ and $\lim_{n \rightarrow \infty} g(n) = + \infty$, the following limits hold
\begin{eqnarray*}
\lim_{n \rightarrow \infty} \frac{2(f(n) + 1)}{g(n)} & = & 2 \lim_{n \rightarrow \infty} \frac{f(n)}{g(n)} + \lim_{n \rightarrow \infty} \frac{2}{g(n)} = 0 \\
\lim_{n \rightarrow \infty} \frac{2 \sum_{j = 0}^{i} (f(n_j) + 1) + 2(f(n) + 1)}{g(n)} & = & 2 \lim_{n \rightarrow \infty} \frac{f(n)}{g(n)} + \lim_{n \rightarrow \infty} \frac{C + 2}{g(n)} = 0
\end{eqnarray*}
where $C = 2 \sum_{j = 0}^{i} (f(n_j) + 1)$, which ensure that $n_0$ and $n_{i+1}$ are well-defined.

For each $k, i \in \mathbb{N}$, consider the set of sub-languages of words of length $n_i$ decided by the machine $M_k \in \mathcal{C}$ using any possible advice $\alpha$ of size $f$, i.e.,
$$
\mathcal{L}^{n_i}_k = \left\{ L(M_k / \alpha)^{n_i} : \alpha \text{ is an advice of size $f$} \right\}.
$$
Consider the diagonal $\mathcal{D} = \big( \mathcal{L}^{n_i}_i \big)_{i \in \mathbb{N}}$ of the set $\big( \mathcal{L}^{n_i}_k \big)_{k, n_i \in \mathbb{N}}$. Since there are at most $2^{f(n_i)}$ advice strings of length $f(n_i)$, it follows that $| \mathcal{L}^{n_i}_i | \leq 2^{f(n_i)}$. Using a similar construction as in the proof of Theorem~\ref{th:diag}, we can define by induction a sub-language $A^{n_i}_{f(n_i) + 1} \subseteq \Sigma^{n_i}$ such that $A^{n_i}_{f(n_i) + 1} \not \in \mathcal{L}^{n_i}_i$. Then, consider the language
$$
A = \bigcup_{i \in \mathbb{N}} A^{n_i}_{f(n_i) + 1} = \bigcup_{i \in \mathbb{N}} A^{n_i}.
$$
Once again, a similar argument as in the proof of Theorem~\ref{th:diag} ensures that $A \not \in \mathcal{C} / f$. Since $\mathcal{C} / f^* \subseteq \mathcal{C} / f$, it follows that $A \not \in \mathcal{C} / f^*$.

We now show that $A \in \mathcal{C} / g^*$. Recall that, by construction, $A^{n_i} \subseteq \{ b(j) : 0 \leq j \leq f(n_i) \}$. Consider the word homomorphism $h : \Sigma^* \rightarrow \Sigma^*$ induced by the mapping $0 \mapsto 00$ and $1 \mapsto 11$, and define the symbol $\# = 01$. 
For each $i \in \mathbb{N}$, consider the encoding $\beta^{n_i}_{0} \beta^{n_i}_{1} \cdots \beta^{n_i}_{f(n_i)}$ of $A^{n_i}$ given by
$$
\beta^{n_i}_j = 
\begin{cases}
1 & \text{if } b(j) \in A^{n_i}, \\
0 & \text{otherwise,}
\end{cases}
$$
for all $0 \leq j \leq f(n_i)$, and let $\beta(n_i) = h(\beta^{n_i}_{0} \beta^{n_i}_{1} \cdots \beta^{n_i}_{f(n_i)})$.
Note that $| \beta(n_i) | = 2( f(n_i) + 1 )$. Now, consider the advice function $\alpha : \mathbb{N} \rightarrow \Sigma^*$ given by the concatenation of the encodings of the successive $A^{n_i}$ separated by symbols $\#$. Formally, 
$$
\alpha(n) = 
\begin{cases}
\beta(n_0) \# \beta(n_1) \# \cdots \# \beta(n_i) & \text{ if $n = n_i$ for some $i \geq 0$} \\
\beta(n_0) \# \beta(n_1) \# \cdots \# \beta(n_i) \# & \text{ if $n_i < n < n_{i+1}$}.
\end{cases}
$$
Note that $| \alpha(n) | = 2 \sum_{j=0}^i (f(n_i) + 1) \leq g(n_i) \leq g(n)$ for $n_i \leq n < n_{i+1}$, and that $\alpha$ satisfies the prefix property: $m \leq n$ implies $| \alpha(m) | \leq | \alpha(n) |$. If necessary, the advice strings can be extended by dummy symbols $10$ in order to achieve the equality $| \alpha(n) | = g(n)$, for all $n \geq 0$ (assuming without loss of generality that $g(n)$ is even). Now, consider the machine with advice $M / \alpha$ which, on every input $w$ of length $n$, first reads its advice string $\alpha(n)$ up to the end. If the last symbol of $\alpha(n)$ is $\#$, then it means that $|w| \neq n_i$ for all $i \geq 0$, and the machine rejects $w$. Otherwise, the input is of length $n_i$ for some $i \geq 0$. Hence, the machine moves its advice head back up to the last $\#$ symbol, and then moves one step to the right. At this point, the advice head points at the beginning of the advice substring $\beta(n_i)$. Then, the machine decodes $\beta^{n_i}_{0} \beta^{n_i}_{1} \cdots \beta^{n_i}_{f(n_i)}$ from $\beta(n_i)$ by removing one out of two bits. Next, as in the proof of Theorem~\ref{th:diag}, the machine moves its advice head up to the $j$-th bit $\beta^{n_i}_j$, where $j = b^{-1}(w)$, if this bit exists (note that $j < 2^{n_i}$ and $| \beta(n_i) | = f(n_i) + 1$), and accepts $w$ if and only if $\beta^{n_i}_j = 1$. These instructions can be computed in time $O(2 \sum_{j = 0}^{i} (f(n_j) + 1) + {n_i})$. It follows that $ w \in L(M / \alpha)^{n_i}$ iff $w \in A^{n_i}$. Thus $L(M / \alpha)^{n_i} = A^{n_i}$, for all $i \in \mathbb{N}$, and hence
$$
L(M / \alpha) = \bigcup_{i \in \mathbb{N}} L(M / \alpha)^{n_i} = \bigcup_{i \in \mathbb{N}} A^{n_i}  = A
$$
Therefore, $A \in \mathcal{C} / g^*$.

Finally, the two properties $A \not \in \mathcal{C} / f^*$ and $A \in \mathcal{C} / g^*$ imply that $\mathcal{C} / f^* \subsetneq \mathcal{C} / g^*$.
\end{proof}

% PARAGRAPH %
The separability between classes of analog, evolving, and stochastic recurrent neural networks using real weights, evolving weights, and probabilities of different Kolmogorov complexities, respectively, can now be obtained. 
%We can use both theorems to obtain many kind of strict hierarchies inside $\mathbf{P/poly}$ or $\mathbf{BPP/\log^*}$, much more general that what was stated in the articles of Siegelmann.

\begin{corollary}
\label{cor:nn_hierarchy}
Let $\mathcal{F}$ and $\mathcal{G}$ be two classes of reasonable advice bounds such that, there is a $g \in \mathcal{G}$, such that for every $f \in \mathcal{F}$,  $f \in o(g)$ and $\underset{n \rightarrow \infty}{\lim} g(n) = + \infty$. Then
\begin{enumerate}[label=(\roman*),itemsep=0pt]
\item \label{point1} $\mathbf{ANN} \left[ K_{\mathrm{poly}}^\mathcal{F}, \mathrm{poly} \right] \subsetneq \mathbf{ANN} \left[ K_{\mathrm{poly}}^\mathcal{G}, \mathrm{poly} \right]$
\item \label{point2} $\mathbf{ENN} \left[ \bar K_{\mathrm{poly}}^\mathcal{F}, \mathrm{poly} \right] \subsetneq \mathbf{ENN} \left[ \bar K_{\mathrm{poly}}^\mathcal{G}, \mathrm{poly} \right]$
\item \label{point3} $\mathbf{SNN} \left[ K_{\mathrm{poly}}^\mathcal{F}, \mathrm{poly} \right] \subsetneq \mathbf{SNN} \left[ K_{\mathrm{poly}}^\mathcal{G}, \mathrm{poly} \right]$
\end{enumerate}
\end{corollary}

\begin{proof}
\ref{point1} and \ref{point2}: Theorem~\ref{hierarchy_thm1} shows that
\begin{eqnarray*}
\mathbf{P} / \mathcal{F}^* & = & \mathbf{ANN} \left[ K_{\mathrm{poly}}^\mathcal{F}, \mathrm{poly} \right] = \mathbf{ENN} \left[ \bar K_{\mathrm{poly}}^\mathcal{F}, \mathrm{poly} \right] \text{ ~and~ } \\
\mathbf{P} / \mathcal{G}^* & = & \mathbf{ANN} \left[ K_{\mathrm{poly}}^\mathcal{G}, \mathrm{poly} \right] = \mathbf{ENN} \left[ \bar K_{\mathrm{poly}}^\mathcal{G}, \mathrm{poly} \right] .
\end{eqnarray*}
In addition, Theorem~\ref{th:diag_2} ensures that
\begin{eqnarray*}
\mathbf{P} / \mathcal{F}^* \subsetneq \mathbf{P} /\mathcal{G}^*
%\bigcup_{f \in \mathcal{F}} \mathbf{P} / f^* \subsetneq  \mathbf{P} / g^* \subset \mathbf{P} /\mathcal{G}^* .
\end{eqnarray*}
The strict inclusions of Points \ref{point1} and \ref{point2} directly follow.

\ref{point3}: Theorem~\ref{hierarchy_thm2} states that
\begin{eqnarray*}
\mathbf{BPP} / (\mathcal{F} \circ \log)^* & = & \mathbf{SNN} \left[ K_{\mathrm{poly}}^\mathcal{F}, \mathrm{poly} \right] \\
\mathbf{BPP} / (\mathcal{G} \circ \log)^* & = & \mathbf{SNN} \left[ K_{\mathrm{poly}}^\mathcal{G}, \mathrm{poly} \right] .
\end{eqnarray*}
In addition, note that if $f \in \mathcal{F}$ and $g \in \mathcal{G}$ satisfy the hypotheses of Theorem~\ref{th:diag_2}, then so do $f \circ l \in \mathcal{F} \circ \log$ and $g \circ l \in \mathcal{G} \circ \log$, for all $l \in \log$. Hence, Theorem~\ref{th:diag_2} ensures that 
\begin{eqnarray*}
\mathbf{BPP} / (\mathcal{F} \circ \log)^* \subsetneq \mathbf{BPP} / (\mathcal{G} \circ \log)^* .
\end{eqnarray*}
The strict inclusion of Point \ref{point3} ensues.
\end{proof}

% PARAGRAPH %
Finally, Corollary~\ref{cor:nn_hierarchy} provides a way to construct infinite hierarchies of classes of analog, evolving and stochastic neural networks based on the Kolmogorov complexity of their underlying weights and probabilities, respectively. The hierarchies of analog and evolving networks are located between $\mathbf{P}$ and $\mathbf{P/poly}$. Those of stochastic networks lie between $\mathbf{BPP}$ and $\mathbf{BPP/log^*}$.

% PARAGRAPH %
For instance, define $\mathcal{F}_i = O \left( \log(n)^i \right)$, for all $i \in \mathbb{N}$. Each $\mathcal{F}_i$ satisfies the three conditions for being a class of reasonable advice bounds (note that the sub-linearity $\log(n)^i$ is satisfied for $n$ sufficiently large). By Corollary~\ref{cor:nn_hierarchy}, the sequence of classes $(\mathcal{F}_i)_{i \in \mathbb{N}}$ induces the following infinite strict hierarchies of classes of neural networks:
\begin{align*}
\mathbf{ANN} \left[ K_{\mathrm{poly}}^{\mathcal{F}_0}, \mathrm{poly} \right] && \subsetneq && \mathbf{ANN} \left[ K_{\mathrm{poly}}^{\mathcal{F}_1}, \mathrm{poly} \right] && \subsetneq && \mathbf{ANN} \left[ K_{\mathrm{poly}}^{\mathcal{F}_2}, \mathrm{poly} \right] && \subsetneq && \cdots  \\
\mathbf{ENN} \left[ \bar K_{\mathrm{poly}}^{\mathcal{F}_0}, \mathrm{poly} \right] && \subsetneq && \mathbf{ENN} \left[ \bar K_{\mathrm{poly}}^{\mathcal{F}_1}, \mathrm{poly} \right] && \subsetneq && \mathbf{ENN} \left[ \bar K_{\mathrm{poly}}^{\mathcal{F}_2}, \mathrm{poly} \right] && \subsetneq && \cdots  \\
\mathbf{SNN} \left[ K_{\mathrm{poly}}^{\mathcal{F}_0}, \mathrm{poly} \right] && \subsetneq && \mathbf{SNN} \left[ K_{\mathrm{poly}}^{\mathcal{F}_1}, \mathrm{poly} \right] && \subsetneq && \mathbf{SNN} \left[ K_{\mathrm{poly}}^{\mathcal{F}_2}, \mathrm{poly} \right] && \subsetneq && \cdots 
\end{align*}

% PARAGRAPH %
We provide another example of hierarchy for stochastic networks only. In this case, it can be noticed that the third condition for being a class of reasonable advice bounds can be relaxed: we only need that $\mathcal{F} \circ \log$ is a class of reasonable advice bounds and that the functions of $\mathcal{F}$ are bounded by $n$. Accordingly, consider some infinite sequence of rational numbers $(q_i)_{i \in \mathbb{N}}$ such that $0 < q_i < 1$ and $q_i < q_{i+1}$, for all $i \in \mathbb{N}$, and define $\mathcal{F}_i = O(n^{q_i})$, for all $i \in \mathbb{N}$. Each $\mathcal{F}_i$ satisfies the required conditions. By Corollary~\ref{cor:nn_hierarchy}, the sequence of classes $(\mathcal{F}_i)_{i \in \mathbb{N}}$ induces the following infinite strict hierarchies of classes of neural networks:
\begin{align*}
\mathbf{SNN} \left[ K_{\mathrm{poly}}^{\mathcal{F}_0}, \mathrm{poly} \right] ~\subsetneq~ 
\mathbf{SNN} \left[ K_{\mathrm{poly}}^{\mathcal{F}_1}, \mathrm{poly} \right] ~\subsetneq~ 
\mathbf{SNN} \left[ K_{\mathrm{poly}}^{\mathcal{F}_2}, \mathrm{poly} \right] ~\subsetneq~ \cdots
\end{align*}

% The class $O(\log(n)^i)$ is a class of reasonable advice bounds, since it satisfies the three conditions when $i$ is an integer larger than $0$. Moreover, $\log(n)^i \in o(\log(n)^{i+1})$, hence by Corollary~\ref{cor:nn_hierarchy}, 
% $$\mathbf{ANN} \left[ K_{\mathrm{poly}}^{O(\log(n)^i)}, \mathrm{poly} \right] \subsetneq \mathbf{ANN} \left[ K_{\mathrm{poly}}^{O(\log(n)^{i+1})}, \mathrm{poly} \right].$$
% We have thus established a strict infinite hierarchy in between $\mathbf{P/\log} = \mathbf{ANN} \left[ K_{\mathrm{poly}}^{O(\log(n))}, \mathrm{poly} \right]$ and $\mathbf{P/poly} = \mathbf{ANN} \left[ \sigma^{\omega} , \mathrm{poly} \right]$.

% TODO: should we consider the log iterate hierarchy ?

% For the case of stochastic networks, in Corollary~\ref{cor:nn_hierarchy}, we only need $\mathcal{F} \circ \log$ to be a class of reasonable advice bounds and the functions of $\mathcal{F}$ to be bounded by $n$. If we consider $O(n^{r})$ with $r$ a rational less than $1$, all conditions are satisfied. As a consequence, for two rationals $r < r'<1$, then  $$\mathbf{SNN} \left[ K_{\mathrm{poly}}^{O(n^{r})}, \mathrm{poly} \right] \subsetneq \mathbf{SNN} \left[ K_{\mathrm{poly}}^{O(n^{r'})}, \mathrm{poly} \right].$$
% Again, we obtain a strict infinite hierarchy of classes in between $\mathbf{BPP/log^*} = \mathbf{SNN} \left[ \sigma^{\omega}, \mathrm{poly} \right]$ and $\mathbf{SNN} \left[ K_{\mathrm{poly}}^{O(1)}, \mathrm{poly} \right] = \mathbf{BPP}$.

\section{Conclusion}
\label{sec:discussion}

% PARRAGRAPH %
We provided a refined characterization of the super-Turing computational power of analog, evolving, and stochastic recurrent neural networks based on the Kolmogorov complexity of their underlying real weights, rational weights, and real probabilities, respectively. For the two former models, infinite hierarchies of classes of analog and evolving networks lying between $\mathbf{P}$ and $\mathbf{P/poly}$ have been obtained. For the latter model, an infinite hierarchy of classes of stochastic networks located between $\mathbf{BPP}$ and $\mathbf{BPP/log^*}$ has been achieved. Beyond proving the existence and providing examples of such hierarchies, Corollary~\ref{cor:nn_hierarchy} establishes a generic way of constructing them based on classes of functions satisfying the reasonable advice bounds conditions.

% PARRAGRAPH %
This work is an extension of the study from Balc\'azar et al.~\cite{SiegelmannEtAl97} about a Kolmogorov-based hierarchization of analog neural networks. In particular, the proof of Theorem~\ref{hierarchy_thm1} draws heavily on their Theorem 6.2~\cite{SiegelmannEtAl97}. In our paper however, we adopted a relatively different approach guided by the intention of keeping the computational relationships between recurrent neural networks and Turing machines with advice as explicit as possible. In this regard, Propositions~\ref{ANN_prop}, \ref{ENN_prop} and \ref{SNN_prop} characterize precisely the connections between the real weights, evolving weights, or real probabilities of the networks, and the advices of different lengths of the corresponding Turing machines. On the contrary, the study of Balc\'azar et al.~keeps these relationships somewhat hidden, by referring to an alternative model of computation: the Turing machines with tally oracles. Another difference between the two works is that our separability results (Theorem~\ref{th:diag} and \ref{th:diag_2}) are achieved by means of a diagonalization argument holding for any non-uniform complexity classes, which is a result of specific interest per se. In particular, our method does not rely on the existence of reals of high Kolmogorov complexity. A last difference is that our conditions for classes of reasonable advice bounds are slightly weaker than theirs.

% PARRAGRAPH %
The computational equivalence between stochastic neural networks and Turing machines with advice relies on the results from Siegelmann~\cite{Siegelmann99b}. Our Proposition~\ref{SNN_prop} is fairly inspired by their Lemma 6.3~\cite{Siegelmann99b}. Yet once again, while their latter Lemma concerns the computational equivalence between two models of bounded error Turing machines, our Proposition~\ref{SNN_prop} describes the explicit relationship between stochastic neural networks and Turing machines with logarithmic advice. Less importantly, our appeal to union bound arguments allows for technical simplifications of the arguments presented in their Lemma.

% PARRAGRAPH %
The main message conveyed by our study is twofold: (1) the complexity of the real or evolving weights does matter for the computational power of analog and evolving neural networks; (2) the complexity of the source of randomness does also play a role in the capabilities of stochastic neural networks. These theoretical considerations contrast with the practical research path about approximate computing, which concerns the plethora of approximation techniques -- among which precision scaling -- that could sometimes lead to disproportionate gains in efficiency of the models~\cite{Mittal16}. In our context, the less compressible the weights or the source of stochasticity, the more information they contain, and in turn, the more powerful the neural networks employing them.

% PARRAGRAPH %
For future work, hierarchies based on different notions than the Kolmogorov complexity could be envisioned. In addition, 
%while the universality of echo state networks has been thoroughly investigated from the perspective of functional analysis~\cite{GrigoryevaOrtega18a,GrigoryevaOrtega18b,GononOrtega20,GononOrtega21}, the topic 
the computational universality of echo state networks could be studied from a probabilistic perspective. Given some echo state network $\mathcal{N}$ complying with well-suited conditions on its reservoir, and given some computable function $f$, is it possible to find output weights such that $\mathcal{N}$ computes like $f$ with high probability?

% PARAGRAPH %
Finally, the proposed study intends to bridge some gaps and present a unified view of the refined capabilities of analog, evolving and stochastic recurrent neural networks. The debatable question of the exploitability of the super-Turing computational power of neural networks lies beyond the scope of this paper, and fits within the philosophical approach to hypercomputation~\cite{Siegelmann99,Copeland02,Copeland04}. Nevertheless, we believe that the proposed study could contribute to the progress of analog computation~\cite{BournezPouly21}.

\section*{Acknowledgements}

% PARAGRAPH %
This research was partially supported by Czech Science Foundation, grant AppNeCo \#GA22-02067S, institutional support RVO: 67985807.

%\pagebreak
\bibliographystyle{plain}
\bibliography{biblio}

\end{document}